\documentclass[a4paper,12pt]{article}

\usepackage{stmaryrd}
\usepackage{amsmath,amsthm}
\usepackage{amssymb}
\usepackage{latexsym}
\usepackage{authblk, mathabx}
\usepackage{tcolorbox}
\usepackage{caption}
\usepackage{subcaption}
\usepackage{float}

\usepackage[]{hyperref}
\usepackage{natbib}
\usepackage[margin=20mm]{geometry}

\newtheorem{theorem}{Theorem}[section]
\newtheorem{proposition}[theorem]{Proposition}
\newtheorem{lemma}[theorem]{Lemma}
\newtheorem{corollary}[theorem]{Corollary} 
\newtheorem{remark}[theorem]{Remark}
\newtheorem{definition}[theorem]{Definition}

\newtheorem{assumption}[theorem]{Assumption}

\newcommand{\bX}{\mathbf{X}}

\newcommand{\bV}{\mathbf{V}}

\newcommand{\bM}{\mathbf{M}}

\newcommand{\bx}{\mathbf{x}}
\newcommand{\bu}{\mathbf{u}}
\newcommand{\by}{\mathbf{y}}

\begin{document}

\title{On short-time behavior of implied volatility in a market model with indexes \footnote{Thai Nguyen acknowledges the support of the Natural Sciences and Engineering Research Council of Canada [RGPIN-2021-02594]. We thank Thomas Bernhardt for useful discussions regarding fractional volatility models.}}
\author[1]{Huy N. Chau}
\author[2]{Duy Nguyen}
\author[3]{Thai Nguyen}

\affil[1]{Department of Mathematics, University of Manchester}
\affil[2]{Department of Mathematics, Marist College}
\affil[3]{\'Ecole d'actuariat, Universit\'e Laval}
\maketitle

\begin{abstract}
	This paper investigates short-term behaviors of implied volatility of derivatives written on indexes in equity markets when the index processes are constructed by using a ranking procedure. Even in simple market settings where stock prices follow geometric Brownian motion dynamics, the ranking mechanism can produce the observed term structure of at-the-money (ATM) implied volatility skew for equity indexes. Our proposed models showcase the ability to reconcile two seemingly contradictory features found in empirical data from equity markets: the long memory of volatilities and the power law of ATM skews. Furthermore, the models allow for the capture of a new phenomenon termed the quasi-blow-up phenomenon.
\end{abstract}

\section{Introduction}
The volatility modeling literature has introduced a variety of models, ranging from the well-known Black-Scholes model pioneered in \citet{black1973pricing}, where volatility is considered constant, to local/stochastic volatility models,  all geared towards capturing the intricacies of reality. Lately, there has been a growing adoption of fractional Brownian motions in volatility modelling. The existence of volatility persistence is well-documented, with seminal analyses by \citet{ding1993long},  \citet{andersen1997intraday}.
\citet{comte1998long} presented a stochastic volatility model wherein the volatility process is governed by the exponential of a fractional Brownian motion with a Hurst exponent $H\in(1/2,1)$. \citet{andersen2003modeling} found that a simple long-memory Gaussian vector autoregression for the logarithmic daily realized volatilities generally produces superior forecasts. Subsequently, an extensive body of literature has expanded on these fractional volatility models, exemplified by works of  \citet{comte2012affine}, \citet{rosenbaum2008estimation} among many others.

In a different vein, \citet{gatheral2018volatility} conducted an innovative study by estimating volatilities from high-frequency data, showing that spot volatilities exhibit rough behaviour across numerous financial assets. Their findings suggested that log-volatility can be effectively modelled by a fractional Brownian motion with a Hurst exponent of order $0.1$. The findings were later reaffirmed in
\citet{livieri2018rough} by using   option prices. In the context of fractional  volatility models, \citet{fukasawa2022consistent}  constructed a quasi-likelihood estimator and applied it to realized volatility time series. Their empirical studies for major stock indices indicate that the Hurst exponents are consistently less than $0.5$. 
\begin{figure}[h!tbp]
	\centering
	\begin{subfigure}{0.45\textwidth}
		\includegraphics[width=\textwidth]{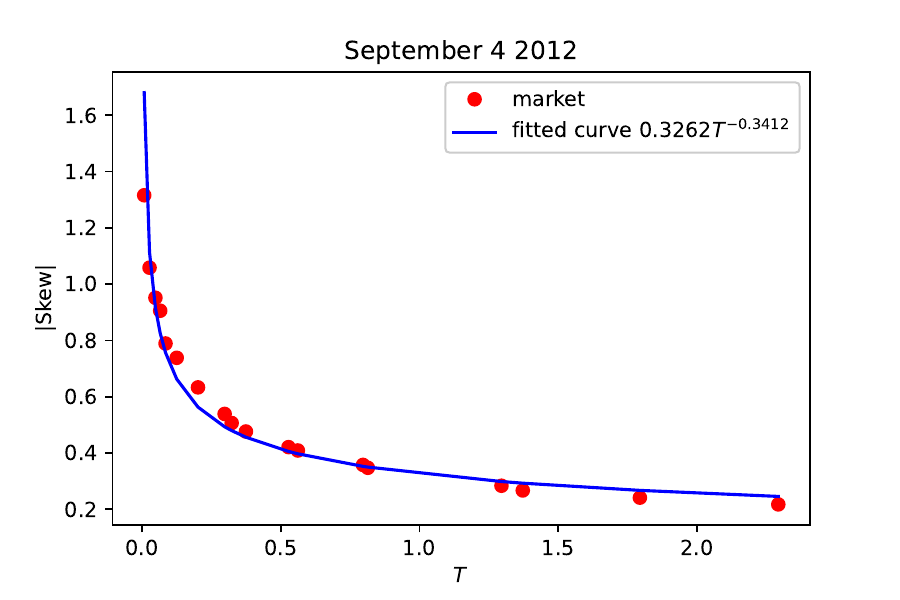}
		\caption{Sep-04-2012.}
		\label{fig:first}
	\end{subfigure}
	\begin{subfigure}{0.45\textwidth}
		\includegraphics[width=\textwidth]{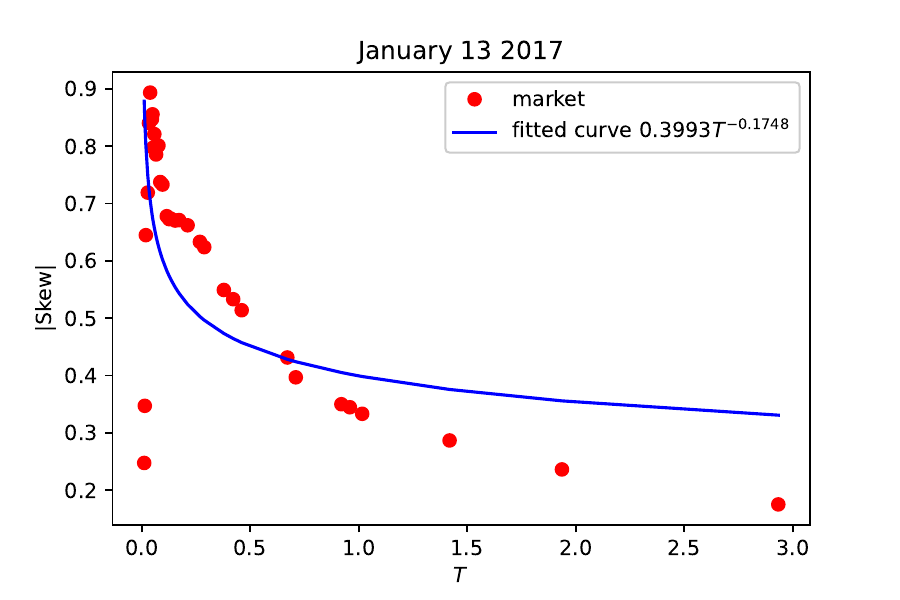}
		\caption{Jan-13-2017.}
		\label{fig:second}
	\end{subfigure}
	\caption{\small The absolute of the ATM implied volatility skew of SP500 options is plotted as a function of the maturity $T$. The power-law term structure of skew (i.e., $|Skew| \sim cT^{-\alpha}$, $\alpha\in (0,1/2)$)  aligns remarkably with SP 500 option data.  We reproduce Figure 9 of \citet{romer2022empirical}. Data is given by Optionmetric.}
	\label{fig:quasi_emp}
\end{figure}
The so-called rough volatility models have proven instrumental in capturing the power-law term structure of ATM implied volatility skew (see Figure \ref{fig:quasi_emp}) that local or stochastic
volatility models typically fail to generate, see \citet{bayer2016pricing}, \citet{fukasawa2017short}, and \citet{bayer2019short}. 
In addition to the empirical evidence, there exists a substantial theoretical foundation supporting rough volatility models,
see \citet{fukasawa2021volatility}, \citet{jaisson2016rough}, \citet{el2018microstructural}.
For a comprehensive exploration of related studies and theoretical underpinnings, we refer to \citet{funahashi2017does},  \citet{el2019characteristic},  \citet{bayer2020regularity},  \citet{forde2021small},  \citet{forde2017asymptotics},  \citet{friz2022short}, and  \citet{friz2021precise}, among others, though this list is by no means exhaustive.


The debate on the nature of dependence in volatility, whether short-range or long-range, has perennially held significance in volatility modeling, as emphasized in  \citet{cont2007volatility}.  \citet{romer2022empirical} documented that the SPX and VIX option markets can be effectively reconciled with classical two-factor volatility models without roughness and jumps. \citet{guyon2022does} conducted empirical investigations into the term structure of the ATM skew of equity indexes, revealing a degradation of the power-law fit with two parameters for short maturities.
Delving more into statistical evidence,  \citet{cont2022rough} examined such evidence for the use of fractional processes with $H < 0.5$ using the concept of normalized $p$-th variation in a framework with microstructure noises. Their results show that although the spot volatility follows Brownian motion dynamics, the realized volatility exhibits rough behavior with a Hurst index $H < 0.5$. This suggests that the origin of roughness observed in realized volatility time series may lie in microstructure noise. In a recent work,  \citet{shi2022volatility} modeled the log realized volatility by an autoregressive fractionally integrated moving average ARFIMA$(1,d,0)$ process where $d>0$ indicates long memory and $d<0$ implies antipersistency. The author applied four estimation methods and explained that all methods have finite sample problems, precluding definitive conclusions about the data-generating processes.

As highlighted in \citet{rogers2023things}, simpler alternative models exist that can elucidate certain empirical properties with efficacy comparable to fractional models, particularly at higher timescales such as daily, weekly, and monthly intervals. \citet{abi2019lifting} introduced lifted versions of the Heston model and demonstrated that these lifted models, being Markovian, adeptly fit implied volatilities for short maturities while aligning with the statistical roughness of realized volatilities.  \citet{bennedsen2022decoupling} employed Brownian semistationary processes, incorporating both roughness and persistence in volatility, along with other desirable properties. \citet{guyon2022does} introduced three-parameter shapes, such as time-shifted or capped power laws, which maintain a semblance of power laws for larger maturities but do not blow  up at vanishing maturity.

Implied volatility, as the market's forecast for future volatilities,  plays as a pivotal role in option pricing, see \citet{durrleman2010implied},  \citet{berestycki2004computing},  \citet{gao2014asymptotics}, \citet{fukasawa2011asymptotic}. Implied volatilities from options prices are achieved by inverting the Black-Scholes formula.  The short-term behaviors of implied volatilities have been explored in various models.  \citet{alos2007short} delved into jump-diffusion models,  \citet{forde2009small}, \citet{forde2012small} focused on Heston's model, and  \citet{euch2019short} considered stochastic volatility models, including fractional volatility models. \citet{bayer2019short}, and  \citet{friz2022short} examined the short-term behaviors in rough volatility models by employing the large deviation approach. \citet{pagliarani2017exact} provided the exact Taylor formula for implied volatilities, considering both strike and maturity by approximating the infinitesimal generator of the underlying processes. \citet{barletta2019short} employed a similar approach, deriving closed-form expansions for VIX futures, options, and implied volatilities.



This paper attempts to construct models that reconcile two puzzling empirical findings from equity markets: the long memory of volatilities and the power law of ATM skew. Furthermore,
we construct models explaining the two empirical phenomena showed in Figure \ref{fig:quasi_emp}
in a unified framework. To do this, we introduce a new model incorporating a market index, wherein stock prices undergo ranking based on their values or market capitalizations before aggregating the top-ranked stocks. This approach mirrors the construction methodology commonly observed in most market indexes. Unlike prior studies that modelled indexes or baskets of stocks using weighted sums of stock prices without incorporating ranking procedures, see, e.g., \citet{avellaneda2003application}, \citet{jourdain2012coupling}, \citet{gulisashvili2015implied}, \citet{bayer2014asymptotics}, \citet{friz2023reconstructing}, our model explicitly integrates the ranking procedure. To derive expansions for European index option prices and implied volatilities, we employ the density expansion approach outlined in \citet{euch2019short}. This expansion method proves particularly advantageous in scenarios characterized by short time scales, high dimensional settings (up to $100$ assets, as explored in \citet{bayer2014asymptotics}), and in situations where explicit formulae are unavailable, such as in general stochastic volatility models. Our contributions, limitations, and comparisons to related studies are summarized below.

\begin{itemize}
	\item 
	Our proposed models are capable of generating the power law term structure of the ATM skew for market index options. Notably, even in simple settings with geometric Brownian motions, the ATM skew could exhibit the power law term structure $T^{-0.5}$. Importantly, our models offer a level of simplicity that distinguishes them from those incorporating fractional volatilities. Conventional numerical algorithms such as PDEs remain applicable, underscoring the practicality and feasibility of implementing our models.

	\item  \citet{pigato2019extreme} introduced the following model to explain the power law behavior of the ATM skew
	\begin{equation}\label{eq:pigato}
		dS_t = S_t \sigma_{loc}(S_t)dW_t,
	\end{equation}
	where the local volatility function $\sigma_{\text{loc}}(x) = \sigma_{-}1_{x < R} + \sigma_{+}1_{x \ge R}$ is discontinuous at a \emph{fixed} level $R$. Notably, the ATM skew in this model exhibits a blow-up phenomenon at a rate of $T^{-1/2}$ when $R = S_0$. In comparison, some behaviour of the indexes in our models are similar to that from the process in \eqref{eq:pigato} for very short maturities, because the ranking mechanism does not change the initial configuration of stock prices when time to maturity is small enough.  However, our framework differs from Pigato's model in several key aspects. Firstly, the underlying assets for index options are the index futures, not the indexes themselves. The indexes accommodate discontinuous volatilities due to the ranking mechanism, introducing an additional layer of complexity, and it remains unclear how the volatilities of the index futures are affected in this context. In addition, the indexes are not traded, and therefore constructing hedging strategies requires different arguments.  Secondly, the volatilities of the market indexes are inherently discontinuous at \emph{random} points, adding a stochastic element to the discontinuity. Lastly, the ATM skews in our models experience a blow-up when certain stock prices coincide, a stochastic event occurring at random times. These distinctions highlight the complexities introduced in our setting compared to the model proposed by \citet{pigato2019extreme}. Furthermore,  our techniques with asymptotic density expansion are more general than the use of Fourier transformation in \citet{pigato2019extreme}.

	\item In \citet{guyon2022does}, it is  argued that the ATM skew  seems to follow the power law $T^{-\alpha}$, with $\alpha \in (0,0.5)$, particularly for relatively large maturities, yet refrain from blowing up for vanishing maturities. \citet{guyon2022does} introduced different models with such property, for example, the 3-parameter model derived from simple non-Markovian variance curve models using the Bergomi-Guyon expansion and the simple 4-parameter term-structure model derived from the two-factor Bergomi model with one more parameter for better fits. In the present paper, we introduce the new concept ``quasi-blow-up" to describe this property. Figure \ref{fig:second} provides an  empirical evidence supporting this phenomenon. We show that our proposed model demonstrates the ability to reproduce the new quasi-blow-up phenomenon. More precisely, under certain conditions, the ATM skews blow up when some initial values of stock prices coincide and exhibit quasi-blow-up when initial values of stock prices are close enough. There are differences between our models and the ones in  \citet{guyon2022does}. The first difference is model consistency. On different days, different models of  \citet{guyon2022does} have to be used depending on whether the ATM skews blow up or exhibit quasi-blow-up. Even we know that the 3-parameter model is good for the situation with quasi-blow-up, we also need to recalibrate its parameters for each day, as today calibration may not work for tomorrow data. Unlike \citet{guyon2022does}, there are no parameters to control the power-like shape  in our models, and the quasi-blow-up phenomena are with respect to the initial stock prices. Our models produce simultaneously the two phenomena  without changing parameters and the ATM term structures in our models are time varying.

	\item The ATM skews in \citet{pigato2019extreme} blow up when $R = S_0$ and it could be checked by simple simulation that when $S_0$ is close to $R$, the ATM skew admits power-like shapes. In \citet{fukasawa2021volatility}, it is argued that the local volatility function has to be singular everywhere since the power law is stable in time. This is true if we assume that the ATM skew blows up at every time. However, it could happen that the ATM skew does not blow up but exhibits the  quasi-blow-up phenomena and hence, everywhere singularity is not necessary. To produce the stable power-like term structure, the process $S_t$ in \eqref{eq:pigato}  should stay close to $R$ for all $t$ (which is unrealistic for equity stocks) or more discontinuities need to be introduced in the local volatility function. In the present paper, we choose the latter option and work with market indexes instead of an individual stock. Note that in reality, European options are commonly written for indexes rather than stocks and most empirical studies about blow-up volatility skews focus on index options.

	\item 
	Our model's capacity to capture the quasi-blow-up phenomenon is even more remarkable when log volatilities of stock prices are modeled by fractional Brownian motions with $H \in (0.5,1)$. This continuous-time modeling approach simultaneously accommodates two crucial yet conflicting empirical observations in equity indexes: the persistence or long memory in volatility and the power-law term structure of ATM skew, see Figure \ref{fig:fss1}. This duality underscores the versatility and relevance of our model in reconciling these seemingly contradictory features. 
	\item Constructing an equilibrium model to elucidate the power-law term structure of ATM skews is an intriguing problem. While existing literature, such as \citet{jaisson2016rough} and \citet{el2018microstructural}, offers arguments rooted in microstructure foundations to account for rough volatilities, the question of how rough volatilities manifest in equilibrium remains unanswered. Similarly, comprehending the nuances of the specific 3- and 4-parameter models introduced by \citet{guyon2022does}, the lifted Heston model of \citet{abi2019lifting}, or the use of Brownian semistationary processes in  \citet{bennedsen2022decoupling}, from an equilibrium perspective presents difficulties. In this context, we contribute an additional and simple mechanism to expound upon the power law ATM skews across a broad spectrum of stock price models. This suggests the possibility of constructing equilibrium models featuring the power law ATM term structure through the incorporation of the ranking procedure. The exploration of this avenue is deferred to future studies, promising valuable insights into equilibrium dynamics and the behavior of ATM skews in financial markets.
	\item It is crucial to underscore that we do not claim all the observed blow-up phenomena come from the ranking mechanism. Furthermore, our primary focus does not lie in the calibration aspect, i.e., the fitting of model parameters to financial data. Rather, our emphasis centers on providing a mechanism to elucidate the observed phenomena within equity indexes. The intricate task of calibration, involving the consideration of all individual stocks within the indexes, along with their respective options and the index options, poses a notably high-dimensional challenge. This complex calibration issue is also deferred to future studies.
\end{itemize}
The structure of the paper is outlined as follows. In Section \ref{sec:setting}, we introduce the market model and lay out the main assumptions guiding our analysis. Section \ref{sec:denisty} presents an approximation for the densities of the driving processes inherent in the model. The dynamics of index future prices are examined in Section \ref{sec:future}, while Section \ref{sec:IV} delves into the investigation of implied volatilities. Moving forward, in Section \ref{sec:ex_num}, we provide examples and numerical results to illustrate the practical implications of our model. Proofs supporting our analytical framework are furnished in Section \ref{sec:proofs}, with additional necessary results consolidated in Section \ref{sec:app}.\\

\noindent\textit{Notations.} We use bold letters for vectors, for example, $\bx = (x^1,...,x^n) \in \mathbb{R}^n$. For two vectors $\bx, \by \in \mathbb{R}^n$, their dot product is defined as $\bx \cdot \by = \sum_{i=1}^n x^iy^i$. 
The normal density with mean $\mathbf \mu$ and covariance matrix $\Gamma$ is denoted by $\phi_{\mathbf\mu,\Gamma}(x)$. 
$E[.]$ denotes the expectation. $\mathbf I_d$ denotes the $d\times d$ identity matrix.

\section{Market models with  indexes}\label{sec:setting}
Let $(\Omega, \mathcal{F}, \mathbb{Q})$ be a probability space equipped with a  filtration $(\mathcal{F}_t)_{t \ge 0}$ satisfying the usual assumptions. Assume that interest rate is zero and there are  $n$ stocks $S^1,...,S^n$ whose dynamics under  $\mathbb{Q}$ are given by
\begin{eqnarray}\label{eq:S}
	dS^j_t &=& S^j_t \sum_{k=1}^d \sigma^{jk}_t \left( \rho^{jk} dB^k_t + \sqrt{1-(\rho^{jk})^2} dW^k_t\right) , \ S^j_0 = s^j_0, 
\end{eqnarray}
where $W^k, B^k, k =1,...,d$ are independent $(\mathcal{F}_t)$-Brownian motions and $\rho^{jk} \in [-1,1]$ for $k \in \{ 1,...,d\}, j \in \{1,...,n\}$. Let $(\mathcal{G}_t)_{t \ge 0}$ be a smaller filtration such that $W^k, k = 1,...,d$ are independent of $(\mathcal{G}_t)_{t \ge 0}$,  and $B^{k}, \sigma^{jk}, k \in \{ 1,...,d\}, j \in \{1,...,n\}$ are adapted to $(\mathcal{G}_t)_{t\geq 0}$.  We also assume that $\sigma^{jk}, k \in \{ 1,...,d\}, j \in \{1,...,n\}$ are positive and continuous. Here, $\mathbb{Q}$ is an equivalent local martingale measure for the market. In this section, we work with general volatility processes $\sigma^{jk}$. We may also assume without loss of generality that the initial prices $\mathbf{s}_0 := (s^1_0,...,s^n_0)$ satisfy
\begin{eqnarray}
	s^1_0 \ge s^2_0 \ge \cdots \geq s^n_0.\label{eq:initial}
\end{eqnarray}
Let $Z^j_t = \log (S^j_t), j = 1,...,n$ be the log-price processes. From It\^o's formula, we obtain 
the dynamics of $Z_t^j$ as follows,
\begin{eqnarray}\label{eq:Z}
	dZ^j_t = - \frac{1}{2} \sum_{k=1}^d (\sigma^{jk}_t)^2dt + \sum_{k=1}^d\sigma^{jk}_t\left( \rho^{jk} dB^k_t + \sqrt{1-(\rho^{jk})^2} dW^k_t\right). 
\end{eqnarray}
Define the ranked process as 
$$S^{(1)}_t \ge S^{(2)}_t \ge ...\ge S^{(n)}_t.$$ 

\noindent It is clear that $Z^{(1)}_t \ge Z^{(2)}_t ...\ge Z^{(n)}_t$, where $Z^{(j)}_t = \log(S^{(j)}_t),\ j =1,...,n$.

\begin{remark}[The ranked processes]
	The dynamics of the ranked processes $S^{(j)}, j =1,...,n$ can be computed explicitly. The ranking procedure introduces discontinuity in volatilities and local times in the dynamics of $S^{(j)}, j =1,...,n$.  For example, if $n = 2$ and assume that the two price processes $S^1, S^2$ are pathwise mutually non-degenerate (see Definition 4.1.2 of \citet{fernholz2002stochastic}), the It\^o - Tanaka formula implies that
	$$dS^{(1)}_t = 1_{S^1_t > S^2_t}dS^1_t + 1_{S^2_t > S^1_t}dS^2_t + d\Lambda^{S^1-S^2}_t,$$
	where $$\Lambda^X_t : = \frac{1}{2}\left( |X_t| - |X_0| - \int_0^t \text{sgn}(X_s)dX_s  \right)  $$
	is the local time at $0$ of $X$. A similar representation holds for $S^{(2)}$. We refer to  Chapter 4 of \citet{fernholz2002stochastic} for further computations, and to   \citet{banner2008local} for a general theory of ranked semimartingales. 
\end{remark} 
Let $0 < \overline{n} \le n$ and $w_j, j=1,..., \overline{n}$ be positive constants. Define a market index by  
\begin{eqnarray}\label{eq:defi_index}
	I_t = \sum_{j = 1}^{ \overline{n}} w_jS^{(j)}_t.
\end{eqnarray}
The model could incorporate the case with time-varying index weights by letting $w_j \in \{0,1\}$ and modelling $S^{(j)}_t$ as the \emph{weighted} stock prices. In reality, the index $I_t$  is not tradable. Investors could only trade an index future or an index exchange-traded fund (EFT) tracking the index. In this paper, we consider the price at time $t \le T$ of the index future with maturity $T$, denoted by 
\begin{eqnarray}\label{eq:future}
	F_{t,T} = E[I_T|\mathcal{F}_t].
\end{eqnarray}
For each $j=1,...,n$, we define by $v^j_0(t) : = E\left[ \sum_{k=1}^d (\sigma^{jk}_t)^2\right] $ the forward variance curve  at time 0 
and by
\begin{eqnarray}
	V^j_0(t) = \sqrt{\int_0^t v^j_0(u)du},
\end{eqnarray}
the normalizing quantities. Noting that $B^k, W^k, k =1,...,d$ are independent, we define 
\begin{eqnarray}
	M^j_t &=& \int_{0}^t \sum_{k=1}^d \sigma^{jk}_u\left( \rho^{jk} dB^k_u + \sqrt{1-(\rho^{jk})^2} dW^k_u\right),\nonumber\\ \left\langle M^j\right\rangle _t &=& \int_{0}^{t} \sum_{k=1}^d (\sigma^{jk}_u)^2du,
\end{eqnarray}
and 
\begin{eqnarray}\label{eq:X}
	X^j_t = -\frac{1}{2V^j_0(t)}\left\langle M^j\right\rangle _t + \frac{1}{V^j_0(t)} M^j_t.
\end{eqnarray}
Define $\bM_t := (M^1_t,...,M^n_t), \bV_0(t) := (V^1_0(t),...,V^n_0(t)), \bX_t :=(X^1_t,...,X^n_t)$. The normalizing procedure makes $\bX$ behave as a standard Gaussian process for small $t$. Using these notations, we rewrite 
	\begin{equation}\label{eq:SVX}
		S^{j}_t = e^{Z^j_t} = s^j_0e^{M^j_t - \frac{1}{2}\left\langle M^j\right\rangle _t} = s^j_0 e^{V^j_0(t)X^j_t}, \ j=1,...,n.
\end{equation}
\begin{assumption}\label{assum0} Throughout this paper, we  assume that there are two scenarios for the starting values of stock prices, namely
	\begin{itemize}
		\item[(i)] $s^1_ 0 > s^2_0 > ... > s^n_0 > 0$,
		\item[(ii)] $s^1_ 0 > ...>s^{r-1}_0  = s^r_0 > ... > s^n_0 > 0$ for some $r \in\{ 2,\ldots,n\}$.
	\end{itemize}
\end{assumption}
Assumption \ref{assum0} requires that the starting values of stock prices are different, or there are at most two stocks with the same starting value.	This assumption helps to reduce the number of possible cases in our analysis. Using similar arguments in this paper, it is possible to extend our analysis to the case where Assumption \ref{assum0} is not satisfied.  
\begin{remark}
	It is important to distinguish the condition in Assumption \ref{assum0}(ii) from the collision of stochastic processes. For example, the probability of triple collisions is defined as 
	$$\mathbb{Q}\left(S^{i}_t = S^j_t = S^k_t, \text{ for some } t \ge 0 \right).$$
	In this paper, we fix $t>0$ and the probability of collisions occurring at a fixed time $t$ is zero. In the settings with Brownian motions, sufficient conditions for no triples or no simultaneous collisions at any time are given in \citet{ichiba2010collisions}, \citet{sarantsev2015triple}. 
	For fractional Brownian motions, we refer to \citet{wang2011collision}, \citet{jiang2007collision}, among others.
\end{remark}
Next, we impose some regularity conditions on the volatility and the corresponding martingale  processes. Assumption \ref{assum:main} below is adopted from \citet{euch2019short} to the present multidimensional setting. 
\begin{assumption}\label{assum:main}
	For any $p\ge 1$, 
	\begin{eqnarray} \label{eq:vola}
		\sup_{t \in (0,1)} \frac{1}{t} \left\| \  \int_0^t \sum_{k=1}^d(\sigma^{jk}_u)^2du \right\|_p < \infty. \qquad \sup_{t \in (0,1)}\frac{1}{t} \left\|  \left(    \int_0^t \sum_{k=1}^d (\sigma^{jk}_u)^2du \right)^{-1}  \right\|_p < \infty.
		\label{assum:inte}
	\end{eqnarray}
	The following expansions hold
	\begin{eqnarray}\label{eq:cond_V}
		V^j_0(t) = \sqrt{v^j_0(0)t} + O(t^{1/2 + \zeta^j}),
	\end{eqnarray}
	for some $\zeta^j > 0, j = 1,..,n.$ 
	
	For each $j=1,...,n$, there exists a family of random vectors $$(M^{(0),j}_t,M^{(1),j}_t,M^{(2),j}_t,M^{(3),j}_t)_{t\in [0,1]}$$ such that
	\begin{itemize}
		\item[(i)] for all $t \in [0,1]$, the random vector $\bM^{(0)}_t := (M^{(0),1}_t,...,M^{(0),n}_t)$ has the normal density function $\phi_{\mathbf{\mu},\Gamma}(\bx)$ with mean vector $\mathbf{\mu}$ and  covariance matrix $\Gamma$;
		\item[(ii)] for all $p>0$,
		\begin{eqnarray}\label{assum:moment}
			\sup_{t \in [0,1]} \left\| M^{(k),j}_t\right\|_p < \infty, \qquad k = 1,2,3;
		\end{eqnarray}
		\item[(iii)] for some\footnote{We may choose the same $\varepsilon$ for all $j$.} $H^j \in (0,1), \varepsilon \in (0,H^j/2)$, 
		\begin{eqnarray}
			\lim_{t \to 0} \frac{1}{t^{2H^j+2\varepsilon}} \left\| \frac{M^j_t}{V^j_0(t)} -  M^{(0),j}_t - t^{H^j} M^{(1),j}_t - t^{2H^j} M^{(2),j}_t \right\| _{1 + \varepsilon} &=& 0, \label{assum:M012}\\
			\lim_{t \to 0} \frac{1}{t^{H^j+2\varepsilon}} \left\| \frac{\left\langle M^j \right\rangle_t}{(V^j_0(t))^2} - 1 - t^{H^j}M^{(3),j}_t \right\| _{1+\varepsilon} &=& 0 \label{assum:M3};
		\end{eqnarray}
		\item[(iv)] the following derivatives 
		\begin{eqnarray}
			a^{(k),j}_t(\bx) &=&  \frac{\partial}{\partial x_j} \left\lbrace \left. E\left[ M^{(k),j} _t \right| \bM^{(0)}_t = \bx \right] \phi_{\mathbf{\mu},\Gamma}(\bx)  \right\rbrace, \ j =1,...,n , k = 1,2,3, \label{assum:a}\\
			b^j_t(\bx) &=&  \frac{\partial^2}{\partial x^2_j} \left\lbrace E\left[ \left. M^{(1),j} _t\right| \bM^{(0)}_t = \bx \right] \phi_{\mathbf{\mu},\Gamma}(\bx)  \right\rbrace, \ j=1...,n, \label{assum:b}\\
			c^j_t(\bx) &=&  \frac{\partial^2}{\partial x^2_j} \left\lbrace E\left[ \left. \left| M^{(1),j} _t\right| ^2\right| \bM^{(0)}_t = \bx \right] \phi_{\mathbf{\mu},\Gamma}(\bx) \right\rbrace, \ j=1,...,n ,\label{assum:c}\\
			d^{(1),j,k}_t(\bx) &=& \frac{\partial^2}{\partial x_k\partial x_j} \left\lbrace  E\left[M^{(1),k}_t |\bM^{(0)}_t = \bx \right] E\left[M^{(1),j}_t |\bM^{(0)}_t = \bx \right]
			\phi_{\mathbf{\mu},\Gamma}(\bx)\right\rbrace, \  \label{assum:d}\\
			e^{(1),j,k}_t(\bx)  &=& \frac{\partial^2}{\partial x_k\partial x_j} \left\lbrace  E\left[M^{(1),j}_t |\bM^{(0)}_t = \bx \right]
			\phi_{\mathbf{\mu},\Gamma}(\bx)\right\rbrace, \ j, k= 1,...,n, \label{assum:e}
		\end{eqnarray}
		exist in the Schwartz space. \footnote{The Schwartz space is the function space $\{ f \in C^{\infty}(\mathbb{R}^n,\mathbb{R}): \forall \boldsymbol{\alpha}, \boldsymbol{\beta} \in \mathbb{N}^n, \|f\|_{\boldsymbol{\alpha}, \boldsymbol{\beta}} < \infty \}$, where $C^{\infty}(\mathbb{R}^n,\mathbb{R})$ is the space of smooth functions and  $\|f\|_{\boldsymbol{\alpha},\boldsymbol{\beta}} = \sup_{\bx \in \mathbb{R}^n}  \bx^{\boldsymbol{\alpha}}D^{\boldsymbol{\beta}}f (\bx)$ with the index notation $\bx^{\boldsymbol{\alpha}} = x_1^{\alpha_1}...x^{\alpha_n}_n, D^{\boldsymbol{\alpha}} = D^{\beta_1}_1...D^{\beta_n}_n.$}
	\end{itemize}
\end{assumption}
For simplicity, we assume that for the stock $S^j$, the conditions \eqref{assum:M012}, \eqref{assum:M3} depend only on the corresponding parameter $H^j$. We also need the following assumption.
\begin{assumption}\label{assumption:Gaussian_bound} There exist $0 < T^* \le 1, p > 1/2$ such that
	$$E\left[ e^{p\sum_{j=1}^n \left\langle M^j \right\rangle _{T^*}} \right]  < \infty.$$
\end{assumption}
{Assumption \ref{assumption:Gaussian_bound} is similar to the well-known Novikov condition and is fulfilled for a large class of models, for example, when volatility is a linear function of a Gaussian process.}


\section{Density expansion}\label{sec:denisty}
In general, it is difficult
to find the density $p_t(\bx)$
of $\bX_t$ in a closed form.
As a result, approximation is needed.
In this section, we adopt the characteristic expansion approach from \citet{euch2019short} to find asymptotic distributions of $\bX_t$ in our present multidimensional setting when $H^j\in (0,1)$. 
\begin{theorem}\label{thm:density_expansion}
Let Assumption \ref{assum:main} be in force. Then, the law of $\bX_t = (X^1_t,...,X^n_t)$ admits density $p_t(\bx) := p_t(x_1,...,x_n)$ which satisfies
\begin{eqnarray}\label{eq:q0}
	\sup_{\bx \in \mathbb{R}^n}  \left| p_t(\bx) - q_t(\bx)\right| = \sum_{j=1}^n o(t^{\min\{2H^j,1\}+\varepsilon/2}), 
\end{eqnarray}
as $t \to 0$, where
\begin{eqnarray}
	q_t (\bx) &=& \phi_{\mathbf{\mu},\Gamma}(\bx) - \sum_{j=1}^n t^{H^j} a^{(1),j}_t(\bx) -  \sum_{j=1}^n t^{2H^j} \left( a^{(2),j}_t(\bx) + c^j_t(\bx) \right)   -   \sum_{j=1}^n V^j_0(t)  \frac{\partial}{\partial x_j} \phi_{\mathbf{\mu},\Gamma}(\bx) \nonumber  \\
	&-&  \sum_{j=1}^n  \frac{V^j_0(t) t^{H^j}}{2} \cdot \left( a^{(3),j}_t(\bx) +  b^j_t(\bx)\right)  
	+  \sum_{j=1}^n \frac{(V^j_0(t))^2}{8} \frac{\partial^2}{\partial x^2_j} \phi_{\mathbf{\mu},\Gamma} (\bx) \nonumber \\
	&+& \sum_{1 \le k,j \le n}  t^{H^k + H^j}d^{(1),k,j}_t(\bx) -  \sum_{1 \le k,j\le n}t^{H^j} \frac{V^k_0(t)}{2}e^{(1),k,j}_t(\bx) \nonumber\\
	&+&  \sum_{1 \le k,j \le n} \frac{V^k_0(t)V^j_0(t)}{4} \frac{\partial^2}{\partial x_jx_k}\phi_{\mathbf{\mu},\Gamma}(\bx), \label{eq:q}
\end{eqnarray}
and the functions $a^{j,(i)}_t(\bx),b^j_t(\bx),c^j_t(\bx), d^{(1),j,k}_t(\bx), e^{(1),j,k}_t(\bx), i=1,...,3$ and $k,j=1,...,n$ are defined in \eqref{assum:a}, \eqref{assum:b}, \eqref{assum:c}, \eqref{assum:d}, \eqref{assum:e}, respectively. 
\end{theorem}
The proof of this theorem is given in Subsection \ref{proof:density}. The function $q_t(\bx)$ in $\eqref{eq:q}$ is not necessarily a density and looks complicated at the first glance. However, for the purposes of this paper, it is enough to keep terms with order $t^{\alpha}, \alpha \le 1/2$, and we may ignore many terms in the density expansion. As a first application of Theorem \ref{thm:density_expansion} we have the following estimation:
\begin{lemma}\label{lemma:appro}
Let $q_t(\bx)$ be given in \eqref{eq:q}. Let $f$ be a function such that $f(\bx) \le C|\bx|^m$ for some $C > 0, m \in \mathbb{N}$. For any $A \subset \mathbb{R}^n$, the following estimate holds
$$ E\left[1_{\bX_t \in A} f(\bX_t) \right] = \int_{A} f(\bx) q_t(\bx) d\bx +  \sum_{j=1}^n o(t^{\min\{2H^j,1\} +\varepsilon/4}).$$ 
\end{lemma}
\begin{proof}
We estimate  for any $r \ge 1$ that 
\begin{eqnarray*}
	E[|M^j_t|^{2r}] \le C(r) E[\left\langle M^j \right\rangle_t^{r} ] \le C(r)C^r  t^r ,
\end{eqnarray*}
where $C(r)$ comes from the Burkholder-Davis-Gundy inequality and $C$ is an upper bound from \eqref{eq:vola}. Noting \eqref{eq:cond_V}, we have that
\begin{equation}\label{eq:uniform}
	\sup_{t \in (0,1)}	E[|X^j_t|^{2r}]  \le C'(r),
\end{equation}
for some constant $C'(r) >0$. Let $0 <\eta < \frac{\varepsilon}{4(n+m)} $ be a small number. Next, we decompose 
\begin{eqnarray*}
	\int_{A} |\bx|^m\left| p_t({\bx})-q_t(\bx)\right| d\bx = \int_{A \bigcap \{|\bx| < \frac{1}{t^{\eta}}\}} |\bx|^m\left| p_t({\bx})-q_t(\bx)\right| d\bx +   \int_{A \bigcap \{ |\bx| > \frac{1}{t^{\eta}}\}} |\bx|^m\left| p_t({\bx})-q_t(\bx)\right|d\bx.
\end{eqnarray*}	
Using Theorem \ref{thm:density_expansion}, the first integral is bounded by
\begin{eqnarray*}
	\int_{|\bx| < \frac{1}{t^{\eta}}} |\bx|^m\left| p_t({\bx})-q_t(\bx)\right| d\bx &\le& \frac{1}{t^{m\eta}} \sup_{\bx \in \mathbb{R}^n}  \left| p_t(\bx) - q_t(\bx)\right|    \int_{|\bx| < \frac{1}{t^{\eta}}} d\bx  \\
	&\le& C \frac{1}{t^{m\eta}}  \frac{1}{t^{n\eta}} \sum_{j=1}^n o(t^{\min\{2H^j,1\} + \varepsilon/2})  \\
	&\le& \sum_{j=1}^n o(t^{\min\{2H^j,1\} + \varepsilon/4}). 
\end{eqnarray*}
Fix $r$ such that $r\eta/2 > 1$. We estimate by the H\"older inequality and then by the Markov inequality that
\begin{eqnarray*}
	\int_{|\bx| > \frac{1}{t^{\eta}}} |\bx|^mp_t({\bx})d\bx &=& E[|\bX_t|^m1_{|\bX_t| \ge\frac{1}{t^{\eta}}}] \le (E[|\bX_t|^{2m}])^{1/2}\left( \mathbb{Q}\left( |\bX_t| \ge \frac{1}{t^{\eta}} \right) \right)^{1/2} \\
	&\le& E([|\bX_t|^{2m}])^{1/2}\left( t^{r\eta}E[|\bX_t|^{r}] \right)^{1/2} = O(t), 
\end{eqnarray*}
noting the uniform bound in \eqref{eq:uniform}. We use the formula for $q_t(\bx)$ and Lemma \ref{lemma:tail1} to get that
\begin{eqnarray*}
	\int_{|\bx| > \frac{1}{t^{\eta}}} |\bx|^mq_t({\bx})d\bx = O(t),
\end{eqnarray*}
and the conclusion follows.
\end{proof}

\section{Future prices}\label{sec:future}
Let $\Pi_n$ denote all  permutations of $\{1,2,...,n\}$. For each $\psi_n\in\Pi_n$, define the event
\begin{eqnarray} \label{eq:A}
A^{\psi_n}_T = \{ \omega: S^{\psi_n(1)}_T \ge S^{\psi_n(2)}_T \ge \cdots \ge S^{\psi_n(n)}_T \}, 
\end{eqnarray}
where the notation $S^{\psi_n(k)}$ denotes the $\psi_n(k)$-th stock.  For presentation convenience, we denote 
\begin{eqnarray*}
\nu_k = w_ks^k_0\sqrt{v^k_0(0)} > 0, \quad  k=1,..., n.
\end{eqnarray*}
The following proposition gives an asymptotic representation of the future price when the stocks have distinct initial prices.  
\begin{proposition}\label{pro:future1} 
Let Assumptions \ref{assum:main}, \ref{assumption:Gaussian_bound} be in force.  
If $s^1_ 0 > s^2_0 >...> s^n_0 > 0$ are fixed, then 
\begin{eqnarray*}
	F_{0,T} &=& I_0 + \sum_{k = 1}^{\overline{n}} m^k_1 \sqrt{T} + \sum_{1 \le k  \le \overline{n}, 1\le j \le n} m^{k,j}_2T^{H^j+1/2}  +  \sum_{1 \le k  \le \overline{n}, 1\le j \le n} m^{k,j}_3T^{2H^j+1/2}\\
	&& +  \sum_{1 \le k  \le \overline{n}, 1\le j, \ell \le n} m^{k,j,\ell}_4 T^{H^k + H^j+1/2} + O(T) + \sum_{k=1}^{\overline{n}} O(T^{1/2 + \zeta^k}) + \sqrt{T} \sum_{j=1}^n o(T^{\min\{2H^j,1\} +\varepsilon/4}),
\end{eqnarray*}
where
\begin{eqnarray*}
	m^k_1 &=&  \int_{\mathbb{R}^n} { \nu_kx_k  \phi_{\mathbf{\mu},\Gamma}\left( \bx\right) d\bx}, \\
	m^{k,j}_2 &=& -  \int_{\mathbb{R}^n}  \nu_{k}x_{k}    a^{(1),j}_T(\bx) d\bx,\\
	m^{k,j}_3 &=& -  \int_{\mathbb{R}^n} \nu_{k}x_{k}  \left( \frac{1}{2}a^{(2),j}_T(\bx)  + c^j_T(\bx) \right)  d\bx,\\
	m^{k,j,\ell}_4 &=&  \int_{\mathbb{R}^n} \nu_{k}x_{k}  d^{(1),j,\ell}(\bx)   d\bx,
\end{eqnarray*}
for $k \in \{1,...,\overline{n}\}, j,\ell \in \{1,...,n\}.$
\end{proposition}
Next, we compute the asymptotic
expansion of $F_{0,T}$ 
for the case where exactly two 
stocks have the same initial price.
\begin{proposition}\label{pro:future2}
Let Assumptions \ref{assum:main},  \ref{assumption:Gaussian_bound} be in force. If $s^1_ 0 > ...>s^{r-1}_0  = s^{r}_0 > ... > s^n_0 > 0$ are fixed for some $r \in\{ 2,\ldots,n\}$ , then
\begin{eqnarray*}
	F_{0,T} &=& I_0 +  \sum_{k=1}^{\overline{n}} m^k_5 \sqrt{T}+ \sum_{1\le k \le \overline{n}, 1 \le j \le n} m^{k,j}_6T^{H^j+1/2}+\sum_{1\le k \le \overline{n}, 1 \le j \le n} m^{k,j}_7T^{2H^j + 1/2} \\
	&+&\sum_{1\le k \le \overline{n}, 1 \le j,\ell \le n} m^{k,j,\ell}_8T^{H^j+H^{\ell}+1/2} + O(T) + O(T^{\zeta^{r-1}}) + O(T^{\zeta^{r}}) + \sum_{k=1}^{\overline{n}} O(T^{1/2 + \zeta^k}) \\
	&+& \sqrt{T} \sum_{j=1}^n o(T^{\min\{2H^j,1\} +\varepsilon/4}),
\end{eqnarray*}
where 
\begin{eqnarray*}
	m^k_5&=&  \int_{-\infty}^{\infty} \cdots \int_{-\infty}^{\sqrt{\frac{v^{r-1}_0(0)}{v^r_0(0)}} x_{r-1}} \cdots \int_{-\infty}^{\infty}  { \nu_kx_k  \phi_{\mathbf{\mu},\Gamma}\left( x\right) dx_n ... dx_1} \\
	&& \qquad +  \int_{-\infty}^{\infty} \cdots \int_{-\infty}^{\sqrt{\frac{v^{r}_0(0)}{v^{r-1}_0(0)}} x_{r}} \cdots \int_{-\infty}^{\infty}  { \nu_kx_k  \phi_{\mathbf{\mu},\Gamma}\left( x\right) dx_n ... dx_1}, \\
\end{eqnarray*}
	\begin{eqnarray*}
	m^{k,j}_6&=&  \int_{-\infty}^{\infty} \cdots \int_{-\infty}^{\sqrt{\frac{v^{r-1}_0(0)}{v^r_0(0)}} x_{r-1}} \cdots \int_{-\infty}^{\infty}  \nu_kx_k a^{(1),j}_T(\bx) dx_n ... dx_1 \nonumber \\
	&+&  \int_{-\infty}^{\infty} \cdots \int_{-\infty}^{\sqrt{\frac{v^{r}_0(0)}{v^{r-1}_0(0)}} x_{r}} \cdots \int_{-\infty}^{\infty}  w_ks^k_0x_k a^{(1),j}_T(\bx) dx_n ... dx_1, 	\end{eqnarray*}
\begin{eqnarray*}
	m^{k,j}_7&=&  \int_{-\infty}^{\infty} \cdots \int_{-\infty}^{\sqrt{\frac{v^{r-1}_0(0)}{v^r_0(0)}} x_{r-1}} \cdots \int_{-\infty}^{\infty}   \nu_kx_k     \left( \frac{1}{2}a^{(2),j}_T(\bx)  + c^j_T(\bx) \right)  dx_n ... dx_1\\
	&+&  \int_{-\infty}^{\infty} \cdots \int_{-\infty}^{\sqrt{\frac{v^{r}_0(0)}{v^{r-1}_0(0)}} x_{j}} \cdots \int_{-\infty}^{\infty}   \nu_kx_k     \left( \frac{1}{2}a^{(2),j}_T(\bx)  + c^j_T(\bx) \right)  dx_n ... dx_1,
\end{eqnarray*}
\begin{eqnarray*}
	m^{k,j,\ell}_8&=&   \int_{-\infty}^{\infty} \cdots \int_{-\infty}^{\sqrt{\frac{v^{r-1}_0(0)}{v^r_0(0)}} x_{r-1}} \cdots \int_{-\infty}^{\infty}   \nu_kx_k  d^{(1),j,\ell}(\bx)  dx_n ... dx_1 \nonumber\\
	&+& \int_{-\infty}^{\infty} \cdots \int_{-\infty}^{\sqrt{\frac{v^{r}_0(0)}{v^{r-1}_0(0)}} x_{r}} \cdots \int_{-\infty}^{\infty}   \nu_kx_k  d^{(1),j,\ell}(\bx)  dx_n ... dx_1,\nonumber
\end{eqnarray*}
for $k \in \{1,...,\overline{n}\}, j,\ell \in \{1,...,n\}.$
\end{proposition}
The proofs of Propositions \ref{pro:future1}, \ref{pro:future2} are given in Subsections \ref{proof:future1}, \ref{proof:future2}, respectively. In Propositions \ref{pro:future1}, \ref{pro:future2}, since $\phi_{\mu,\Gamma}$ is a symmetric function, the quantity $m^k_1$ disappears (as seen in Example \ref{ex:2GBMs}) while the quantity $m^k_5$ may be non zero. Therefore, the behaviour of the future prices $F_{t,T}$ are completely different for the two cases in Propositions \ref{pro:future1}, \ref{pro:future2}.

\section{Pricing index call options}\label{sec:IV}
We first recall the Black-Scholes formula for European call options.  
\begin{definition}\label{defi:BS}
The Black-Scholes  price function is denoted by 
\begin{eqnarray}
	C^{BS}(T-t,x,k,\sigma) = N(d_1)x - N(d_2)xe^ke^{-r(T-t)},
\end{eqnarray}
where $k$ is the log strike, $x$ is the spot price at time $t$, $T$ is the maturity of the option, $r$ is the interest rate, and
\begin{eqnarray*}
	d_1 &=& \frac{1}{\sigma\sqrt{T-t}}\left[ - k + \left( r + \frac{\sigma^2}{2} \right) (T-t)  \right],\\
	d_2 &=& d_1 - \sigma\sqrt{T-t},
\end{eqnarray*}
where $N$ is the cumulative distribution function of the standard normal distribution. 
\end{definition}
While the index $I_t$ is not tradable, the future $F_{t,T}$ is tradable and a $\mathbb{Q}$-martingale with $F_{T,T} = I_T$. Here,  $F_{t,T}$ may differ from $I_t$ since $I_t$ is not  necessarily a $\mathbb{Q}$-martingale. Therefore, the ATM strike at time $0$ for the index option is $F_{0,T}$.  Let $ C(T,x,k):= E[(I_T-xe^k)^+]$ be the price at time $0$ of a European call option with the log strike $k$.
\begin{definition}\label{defi_IV}
The implied volatility $\sigma^{IV}: = \sigma^{IV}(T,F_{0,T},k)$ is the solution to the following equation
\begin{eqnarray}\label{eq:implied_vol}
	C^{BS}(T,F_{0,T},k, \sigma^{IV}(T,F_{0,T},k)) = C(T,F_{0,T}, k).
\end{eqnarray}
The ATM skew is defined by
\begin{eqnarray}\label{eq:atmskew}
	ATMskew(T):=\frac{\partial \sigma^{IV}}{\partial k}(T,F_{0,T},k=0).
\end{eqnarray}
\end{definition}
\begin{assumption}\label{assum:C_der_k}
The price $C(T,x,k)$ is continuously differentiable  with respect to $(x,k)$.  
\end{assumption}
{It can be seen that if $I_T$ admits a nice probability density function then Assumption \ref{assum:C_der_k} is satisfied.} We remark that the option prices, implied volatilities, and other related quantities depend on the initial stock price vector $\mathbf{s}_0$ implicitly through the future prices $F_{0,T}$. Below, we write $F_{0,T}(\mathbf{s}_0)$ to emphasize that the future price is a function of the initial stock values. The following result proves the continuity of such quantities with respect to the initial prices when $T$ is fixed. This situation is different from the ones in  Propositions \ref{pro:future1}, \ref{pro:future2} where the initial values for stocks are fixed.
\begin{proposition}\label{lem:lim} Let Assumption \ref{assum:C_der_k} be in force. Fix $T >0$.
\begin{itemize}
	\item[(i)]  For any $\overline{\mathbf{s}}_0 \in \mathbb{R}^n_+$, it holds that $$\lim_{\mathbf{s}_0 \to \overline{\mathbf{s}}_0} F_{0,T}(\mathbf{s}_0) =  F_{0,T}(\overline{\mathbf{s}}_0). $$
	\item[(ii)] For any $\overline{\mathbf{s}}_0 \in \mathbb{R}^n_+$, it holds that $$\lim_{\mathbf{s}_0 \to \overline{\mathbf{s}}_0} \frac{\partial C}{ \partial k}(T,F_{0,T}(\mathbf{s}_0),k=0)  =  \frac{\partial C}{ \partial k} (T,F_{0,T}(\overline{\mathbf{s}}_0),k=0). $$
	\item[(iii)] For any $\overline{\mathbf{s}}_0 \in \mathbb{R}^n_+$, it holds that $$\lim_{\mathbf{s}_0 \to  \overline{\mathbf{s}}_0}\sigma^{IV}(T,F_{0,T}(\mathbf{s}_0),k=0) = \sigma^{IV}(T,F_{0,T}(\overline{\mathbf{s}}_0),k=0).$$
	\item[(iv)] For any $\overline{\mathbf{s}}_0 \in \mathbb{R}^n_+$, $$\lim_{\mathbf{s}_0 \to \overline{\mathbf{s}}_0}  \frac{\partial \sigma^{IV}}{\partial k}(T,F_{0,T}(\mathbf{s}_0),k=0) = \frac{\partial \sigma^{IV}}{\partial k}(T,F_{0,T}(\overline{\mathbf{s}}_0),k=0).$$
\end{itemize}
\end{proposition}
\begin{proof}
Recall that $\Pi_n$ contains all permutations of the set $\{1,2,\ldots,n\}$. From \eqref{eq:future_decomp} and \eqref{eq:A} in Section \ref{proof:future1},  we write 
$$	F_{0,T}(\mathbf{s}_0) = \sum_{\psi_n\in\Pi_n}E\left[ \left(   \sum_{j=1}^{\overline{n}} w_js^{\psi_n(j)}_0e^{V^{\psi_n(j)}_0(T)X^{\psi_n(j)}_T}\right)  1_{s^{\psi_n(1)}_0e^{V^{\psi_n(1)}_0(T)X^{\psi_n(1)}_T} \ge ... \ge s^{\psi_n(n)}_0e^{V^{\psi_n(n)}_0(T)X^{\psi_n(n)}_T}} \right], $$
then $(i)$ follows by the dominated convergence theorem,  noting that for a dominating random variable we can choose $ \sum_{\psi_n\in\Pi_n} \sum_{j=1}^{\overline{n}} w_jMe^{V^{\psi_n(j)}_0(T)X^{\psi_n(j)}_T}$ with some large $M$. Computing the derivative of the option price $C$ w.r.t $k$, we obtain
\begin{equation}\label{eq:C_k}
\frac{\partial C}{\partial k}(T,F_{0,T},k=0) = -F_{0,T}\mathbb{Q}\left( I_T > F_{0,T} \right) = - F_{0,T} \sum_{\psi_n\in\Pi_n}\mathbb{Q}\left( \{I_T > F_{0,T}\} \cap A^{\psi_n}_T\right).
\end{equation}We rewrite \eqref{eq:C_k} as 
\begin{eqnarray*}
	&&\frac{\partial C}{\partial k}(T,F_{0,T}(\mathbf{s}_0),k=0)\\
	&=& - F_{0,T}(\mathbf{s}_0) \sum_{\psi_n\in\Pi_n}E\left[ 1_{\sum_{j=1}^{\overline{n}} w_js^{\psi_n(j)}_0e^{V^{\psi_n(j)}_0(T)X^{\psi_n(j)}_T} > F_{0,T}(\mathbf{s}_0)} 1_{s^{\psi_n(1)}_0e^{V^{\psi_n(1)}_0(T)X^{\psi_n(1)}_T} \ge ... \ge s^{\psi_n(n)}_0e^{V^{\psi_n(n)}_0(T)X^{\psi_n(n)}_T}} \right].
\end{eqnarray*}
The conclusion $(i)$ and the dominated convergence theorem imply $(ii)$. 

We now prove $(iii)$. For a fixed $T>0$, we define the function
\begin{eqnarray}
	f: \mathbb{R}_+ \times \mathbb{R} \times \mathbb{R}_+ &\to& \mathbb{R} \nonumber\\
	f(x,k,\sigma) &=& C^{BS}(T,x,k, \sigma) -  C(T,x,k). \label{eq:f}
\end{eqnarray} 
From \eqref{eq:implied_vol}, the implied volatility $\sigma^{IV} = \sigma^{IV} (T,x,k)$ is the solution to the equation $f(x,k,\sigma) = 0$. Fix a point $(x,0,\sigma)$ such that $f(x,k=0,\sigma) = 0$. The derivative of $f$ w.r.t $\sigma$ at $(x,k=0,\sigma)$  is 
\begin{eqnarray*}
	\frac{\partial f}{\partial \sigma}(x,0,\sigma) &=& \frac{\partial C^{BS}}{\partial \sigma} (T,x,0,\sigma) = N'(d_1) \frac{\partial d_1}{\partial \sigma} x  - N'(d_2) \frac{\partial d_2}{\partial \sigma} x  \\
	&=& N'(d_1) x \frac{1}{2} \sqrt{T}  + N'(d_2)  x \frac{1}{2}  \sqrt{T}> 0.
\end{eqnarray*}
By the implicit function theorem, there exists an open set $U \subset \mathbb{R}_+ \times \mathbb{R}$ containing $(x,k=0)$ and a unique continuously differentiable function $\sigma^{IV}: U \to \mathbb{R}$ such that $f(y,\ell,\sigma^{IV}(y,\ell)) = 0$ and 
\begin{eqnarray} \label{eq:implicit}\left( \frac{\partial \sigma^{IV}}{\partial x}, \frac{\partial \sigma^{IV}}{\partial k} \right) (y,\ell) = - \left(  \frac{\partial f}{\partial \sigma}(y,\ell,\sigma^{IV}(y,\ell)) \right)^{-1} \left(  \frac{\partial f}{\partial x}, \frac{\partial f}{\partial k} \right)(y,\ell,\sigma^{IV}(y,\ell)), \ \forall  (y,\ell) \in U.  
\end{eqnarray} 
From $(i)$, we get that  $\lim_{\mathbf{s}_0 \to \overline{\mathbf{s}}_0}\sigma^{IV}(T,F_{0,T}(\mathbf{s}_0),k=0) = \sigma^{IV}(T,F_{0,T}(\overline{\mathbf{s}}_0),k=0)$, and thus $(iii)$ follows. Finally, the statement $(iv)$ is deduced from \eqref{eq:implicit}.
\end{proof}

\begin{lemma}The ATM skew is computed by
\begin{eqnarray}
	ATMskew(T) = \frac{\sqrt{2\pi}e^{\frac{(\sigma^{IV})^2T}{8}}}{\sqrt{T} } \left( \frac{1}{F_{0,T}} \frac{\partial C}{ \partial k}(T,F_{0,T},k=0)  + N\left( -\frac{1}{2}\sigma^{IV}\sqrt{T} \right)\right).	\label{eq:ATMskew1}
\end{eqnarray}
\end{lemma}
\begin{proof}
From  \eqref{eq:atmskew}, \eqref{eq:implicit}, we get
\begin{eqnarray*} 
	\frac{\partial \sigma^{IV}}{\partial k}(T,F_{0,T},k = 0) &=&  \left(   \frac{\partial C^{BS}}{\partial \sigma}(F_{0,T},k=0,\sigma^{IV}(F_{0,T},k=0))  \right)^{-1}  \\
	&& \times \left(  \frac{\partial C}{\partial k} (F_{0,T},k=0,\sigma^{IV}(F_{0,T},k=0)) -   \frac{\partial C^{BS}}{\partial k} (F_{0,T},k=0,\sigma^{IV}(F_{0,T},k=0)) \right) .  
\end{eqnarray*} 
The vega of the Black-Scholes price is 
\begin{eqnarray}
	\frac{\partial C^{BS}}{\partial \sigma} (F_{0,T},k=0,\sigma^{IV}(F_{0,T},k=0)) &=& N'(d_1) \frac{\partial d_1}{\partial \sigma}  F_{0,T}  - N'(d_2) \frac{\partial d_2}{\partial \sigma} F_{0,T}  \nonumber\\
	&=& F_{0,T}\sqrt{T} \frac{1}{\sqrt{2\pi}}e^{-\frac{(\sigma^{IV})^2T}{8}}, \label{eq:vega}
\end{eqnarray}
where
\begin{eqnarray*}
	d_1 = \frac{1}{\sigma^{IV}\sqrt{T}} \left( - k +  \frac{(\sigma^{IV})^2}{2} T  \right), \qquad 
	d_2 = d_1 - \sigma^{IV}\sqrt{T}. 
\end{eqnarray*}
We also compute
\begin{eqnarray*}
	\frac{\partial C^{BS}}{\partial k}(F_{0,T},k=0,\sigma^{IV}(F_{0,T},k=0))	&=& N'(d_1) \frac{\partial d_1}{\partial k} F_{0,T} - N'(d_2) \frac{\partial d_2}{\partial k}  F_{0,T} -  N(d_2)F_{0,T}\nonumber  \\
	&=&  -  N(d_2)F_{0,T}.
\end{eqnarray*}
The ATM skew formula follows.
\end{proof}

\begin{lemma}\label{lemma:IV_asym} Let Assumptions  \ref{assum:main}, \ref{assumption:Gaussian_bound}  be in force. It holds that
$$\sigma^{IV}(T,F_{0,T},k=0)\sqrt{T} = O(\sqrt{T}).$$
\end{lemma}
\begin{proof}
Using the argument with the Taylor theorem in the proof of Proposition \ref{pro:future1}, we could prove that
\begin{equation}\label{call}
	E[1_{I_T\ge F_{0,T}}(I_T - I_0)] = O(\sqrt{T}).
\end{equation}
Therefore, the ATM call price is
\begin{eqnarray*}
	E[(I_T-F_{0,T})^+] = E[1_{I_T \ge F_{0,T}}(I_T - I_0)] + E[1_{I_T \ge F_{0,T}}(F_{0,T} - I_0)] = O(\sqrt{T}), 
\end{eqnarray*}
from \eqref{call} and Propositions \ref{pro:future1}, \ref{pro:future2}. The ATM implied volatility $\sigma^{IV}(T,F_{0,T}, k =0)$ is the solution of the equation  
$F_{0,T}(N(d_1)-N(d_2)) = E[(I_T-F_{0,T})^+] = O(\sqrt{T})$. We deduce that $N(d_1) = 1/2 + O(\sqrt{T})$ and hence $d_1 = \sigma^{IV}\sqrt{T} = O(\sqrt{T})$. 
\end{proof}
\begin{remark}
Lemma \ref{lemma:IV_asym} is used to study the behaviour of the quantity $\sigma^{IV}\sqrt{T}$ in  \eqref{eq:ATMskew1}, and the order $O(\sqrt{T})$ is enough for our purposes. 
\end{remark}

From empirical studies, it is usually assumed  that the ATM skew is well approximated by a power law function of the time to maturity $T$, see Figure \ref{fig:quasi_emp} (a). Nevertheless, this assumption may need further consideration because the option prices for very short maturities may not be available. In this paper, a phenomenon called ``quasi-blow-up" is introduced, in order to produce the power-like shape of the ATM skew at small maturities as explained in \citet{guyon2022does}. The quasi-blow-up phenomena are with respect to the initial stock prices, which is different from \citet{guyon2022does} where there are parameters to control the power-like shape for small maturities.
\begin{definition}\label{defi:quasi_blow_up}
The ATM skew exhibits a quasi-blow-up phenomenon w.r.t. initial prices $\mathbf{s}_0 \in \mathbb{R}^n_+$ if there is a set $\emptyset \ne \Theta \subset \mathbb{R}^n_+$ such that
\begin{itemize}
	\item[(i)] for $\overline{\mathbf{s}}_0 \in \Theta$, the corresponding ATM skew blows up
	\begin{equation}
		\lim_{T \to 0}  \frac{\partial \sigma^{IV}}{\partial k}(T,F_{0,T}(\overline{\mathbf{s}}_0),k=0)  = \infty;
	\end{equation}
	\item[(ii)] for $\mathbf{s}_0 \in \mathbb{R}^n_+ \setminus \Theta$, the ATM skew does not blow up
	\begin{equation}
		\lim_{T \to 0}  \frac{\partial \sigma^{IV}}{\partial k}(T,F_{0,T}({\mathbf{s}}_0),k=0) < \infty;
	\end{equation}
	\item[(iii)] for any fixed $T >0$, and $\overline{\mathbf{s}}_0 \in \Theta,$ 
	\begin{equation}
		\lim_{ \mathbb{R}^n_+ \setminus \Theta \ni \mathbf{s}_0 \to \overline{\mathbf{s}}_0}  \frac{\partial \sigma^{IV}}{\partial k}(T,F_{0,T}({\mathbf{s}}_0),k=0)  = \frac{\partial \sigma^{IV}}{\partial k}(T,F_{0,T}(\overline{\mathbf{s}}_0),k=0).
	\end{equation}
\end{itemize}
\end{definition}
Now we are ready to state our main results in this section.
\begin{theorem}\label{thm:main}
Let Assumptions \ref{assum:main}, \ref{assumption:Gaussian_bound},  \ref{assum:C_der_k} be in force. Let $q_T(\bx)$ be given in Theorem \ref{thm:density_expansion}.
\begin{itemize}
	\item[(i)] For the case $s^1_ 0 > s^2_0 > ... > s^n_0$, we have
	\begin{eqnarray*}
	\frac{\partial C}{\partial k} (T,F_{0,T},k=0) &=& - F_{0,T}\left(  \int_{D^1} q_T(\bx)  d\bx  + O(\gamma^1(T))\right) ,
	\end{eqnarray*} where
	$D^1, \gamma(T) = \gamma^1(T)$ are given in \eqref{eq:D1}, \eqref{eq:gamma1}.
	\item[(ii)] For the case $s^1_ 0 > ...>s^{r-1}_0  = s^r_0 > ... > s^n_0$ for some $r \in\{ 2,\ldots,n\}$, we have
	\begin{eqnarray*}
	\frac{\partial C}{\partial k}(T,F_{0,T},k=0)  &=& - F_{0,T}\left( \int_{D^{2,1}\cup D^{2,2}} q_T(\bx)  d\bx  + O(T^{\zeta^{r-1}}) + O(T^{\zeta^{r}})
		+ O(\gamma^2(T))\right) ,
	\end{eqnarray*}
	where $D^{2,1}, D^{2,2}, \gamma^2(T)$ are given in \eqref{eq:D21}, \eqref{eq:D22},  \eqref{eq:gamma2}.
\end{itemize}
\end{theorem}
The proof of Theorem \ref{thm:main} is given in Section \ref{proof:thm:main}. Using Theorem \ref{thm:main}, we can study the short time behaviour of the ATM skew by using the formula \eqref{eq:ATMskew1} and Lemma \ref{lemma:IV_asym}. We report some special cases that will be illustrated by concrete examples in Section \ref{sec:ex_num}. 
\begin{corollary}\label{cor_05}
For $j = 1,...,n$,  assume that $H^j \in [0.5,1)$. 
\begin{itemize}
	\item[(i)] Case $s^1_ 0 > s^2_0 > ... > s^n_0$. Assume further that $V^j_0(t) = \sqrt{v^j_0(0)t}$ or $V^j_0(t) = \sqrt{v^j_0(0)t} + O(t^{1/2 + \zeta^j})$ with $\zeta^j > 1/2$. If  $\sum_{k = 1}^{\overline{n}} m^k_1  = 0$,  then  
	$$\frac{\partial C}{\partial k} (T,F_{0,T},k=0) =- F_{0,T}\left( \frac{1}{2} + O(\sqrt{T})\right) ,$$
	and the ATM skew does not blow up.
	\item[(ii)] Case $s^1_ 0 > ...>s^{r-1}_0  = s^r_0 > ... > s^n_0$ for some $r \in\{ 2,\ldots,n\}$: if $\sum_{k=1}^{\overline{n}} m^k_5 \ne 0$, then $\int_{D^{2,1}\cup D^{2,2}}  \phi_{\mathbf{\mu},\Gamma}(\bx) \ne \frac{1}{2}$ and therefore the ATM skew blows up at the rate of $T^{-1/2}$. 
\end{itemize}
This means that the model exhibits the quasi-blow-up phenomenon with the set of initial prices $\Theta = \{ (s^1_0, ...,s^n_0) \in \mathbb{R}^n_+: s^{r-1}_0 = s^r_0 \text{ for some } r = 2,..,n \}$. 
\end{corollary}
\begin{proof}
It suffices to prove $(i)$. When $V^j_0(t) = \sqrt{v^j_0(0)t}$, i.e.  the term $O(t^{1/2 + \zeta^j})$ does not exist in \eqref{eq:cond_V},  the quantity $\gamma^1(T)$ is of order $O(\sqrt{T})$. 
		 
When $V^j_0(t) = \sqrt{v^j_0(0)t} + O(t^{1/2 + \zeta^j})$ with $\zeta^j > 1/2$, we choose $\eta$ small enough such that the condition \eqref{eq:condition_eta} is satisfied and $\zeta^j-\eta > 1/2$, then $\gamma^1(T)= O(\sqrt{T})$.

The conclusion follows from \eqref{eq:ATMskew1} in both cases. 
\end{proof}
Finally, Corollary \ref{cor:<0.5} below
provides several conditions
which guarantee the blow up of the ATM skew when the smallest Hurst parameter $H^j$ is smaller than $1/2$. We emphasize that in this case, the ATM skew could blow up either at the rate $T^{H^j-1/2}$ or surprisingly at the rate $T^{-1/2}$ when two initial stock values coincide. 
\begin{corollary}\label{cor:<0.5}
Assume that at least one Hurst parameter is smaller than $1/2$ and all  the Hurst parameters are different. Let  $H^j < 1/2$ be the smallest among them. 
\begin{itemize}
\item[(i)] Case $s^1_0 > s^2_0 >...>s^n_0$. Assume further that $V^i_0(t) = \sqrt{v^i_0(0)t}$ or $V^i_0(t) = \sqrt{v^i_0(0)t} + O(t^{1/2 + \zeta^i})$ with $\zeta^i > H^i$ for $i=1,..,n$. If  $\sum_{k = 1}^{\overline{n}} m^k_1  = 0$ and $\sum_{k=1}^{\overline{n}}m_2^{kj} \ne 0$ then  
$$ \frac{\partial C}{\partial k} (T,F_{0,T},k=0) = - F_{0,T}\left( \frac{1}{2} + O(T^{H^j})\right) ,$$
and the ATM skew blows up at the rate $T^{H^j-1/2}$.  
\item[(ii)] Case $s^1_ 0 > ...>s^{r-1}_0  = s^r_0 > ... > s^n_0$ for some $r \in\{ 2,\ldots,n\}$: if $\sum_{k=1}^{\overline{n}} m^k_5 \ne 0$, then $\int_{D^{2,1}\cup D^{2,2}}  \phi_{\mathbf{\mu},\Gamma}(\bx) \ne \frac{1}{2}$ and the ATM skew blows up at the rate of $T^{-1/2}$.  
\end{itemize}
\end{corollary}
\begin{proof}
It suffices to prove $(i)$. When $V^j_0(t) = \sqrt{v^j_0(0)t}$, i.e.  the term $O(t^{1/2 + \zeta^j})$ does not exist in \eqref{eq:cond_V},  the quantity $\gamma^1(T)$ is of order $O(T^{H^j})$. 
	
When $V^j_0(t) = \sqrt{v^j_0(0)t} + O(t^{1/2 + \zeta^j})$ with $\zeta^j > H^j$, we choose $\eta$ small enough such that the condition \eqref{eq:condition_eta} is satisfied and $\zeta^j-\eta > H^j$, then $\gamma^1(T)= O(T^{H^j})$.

The conclusion follows from \eqref{eq:ATMskew1}.  
\end{proof}
\begin{remark}
Here, we only present the quasi-blow-up phenomena for a particular day and different configurations for initial stock values. As time varies, we observe the blow up phenomena when the rankings are updated. When stock values change gradually, the quasi-blow-up phenomena are observed, and are expected to persist until stock values are far enough from each others. 
\end{remark}	
\section{Examples and numerical results}\label{sec:ex_num}
In this section, we provide two examples and numerical results to illustrate the practical implications of our models.  
The implementation code is available at \url{https://github.com/nducduy/quasi-blow-up}.
\subsection{Models with  Geometric Brownian motions}\label{ex:2GBMs}
Consider the following model with two geometric Brownian motions
\begin{eqnarray*}\label{eq:S_GBM}
dS^1_t &=& S^1_t \sigma^1dB^1_t, \ S^1_0 = s^1_0,\\
dS^2_t &=& S^2_t\sigma^2dB^2_t, \ S^2_0 = s^2_0,
\end{eqnarray*}
where $B^1_t$ and $B^2_t$ are independent and $\sigma^{1} <  \sigma^{2}$ are positive constants. In this example, 
it can be seen that
\begin{equation*}
v^j_0(t)  = (\sigma^{j})^2, \qquad V^j_0(t) = \sigma^j\sqrt{t}, \qquad M^j_t = \sigma^{j} B^j_t, \qquad \left\langle M^j\right\rangle _t = (\sigma^{j})^2 t, \ j=1,2.
\end{equation*}
Assumptions \ref{assum:main} \ref{assumption:Gaussian_bound}, \ref{assum:C_der_k} are satisfied with 
$$ M^{j,(0)}_t  =  \frac{B^j_t}{\sqrt{t}}, \qquad M^{j,(1)}_t = M^{j,(2)}_t  =  M^{j,(3)}_t = 0,$$
and 
$	a^{(k),j}_t(\bx) = b^j_t(\bx) =  c^j_t(\bx) = d^{(1),jk}_t(\bx) = e^{(1),jk}_t(\bx) = 0, \ 1\le j,k \le 2.$
By Theorem \ref{thm:density_expansion}, the density of $(X^1_t, X^2_t)$ is approximated by
\begin{eqnarray}
q_t(\bx) &=& \phi_{\mathbf{0},\mathbf{I}_{2}}(\bx) - \sum_{j=1}^2\frac{ \sigma^j\sqrt{t}}{2}    \frac{\partial}{\partial x_j}  \phi_{\mathbf{0},\mathbf{I}_2}(\bx)  \nonumber \\
&+&  \sum_{j=1}^2\frac{t(\sigma^j)^2}{8}(x^2_j-1)\phi_{\mathbf{0},\mathbf{I}_2}(\bx) + \sum_{j,k=1}^2 \frac{t\sigma^j\sigma^k}{4} x_kx_j \phi_{\mathbf{0},\mathbf{I}_2}(\bx), \label{eq:app_q} 
\end{eqnarray}
where $ \mathbf{I}_2$ is the $2 \times 2$ identity matrix.  
\begin{remark}
In this simple example, it's easy to compute expansion \eqref{eq:app_q}. Recall that 
\begin{eqnarray*}
X^j_t &=& -\frac{1}{2} \sigma^j\sqrt{t} + \frac{1}{\sqrt{t}}  B^j_t \sim N\left( -\frac{1}{2}\sigma^j\sqrt{t}, 1 \right) .
\end{eqnarray*}
The probability densities of $X^j_t, j = 1,2$ are \begin{eqnarray*}
f^j_t(x_j) &=& \frac{1}{\sqrt{2\pi}} e^{-\frac{1}{2}\left(x_j + \frac{1}{2} \sigma^j \sqrt{t} \right)^2}. 
\end{eqnarray*}
The following expansions hold when $t$ is small 
\begin{eqnarray}
f^j_t(x_j) &=& \frac{1}{\sqrt{2\pi}} e^{-\frac{x_j^2}{2}} \left( 1 - \sigma^j\sqrt{t}  \frac{x_j}{2}  - \frac{1}{8}t (\sigma^j)^2 + \frac{1}{8}x^2_j t(\sigma^j)^2  + O(t^{3/2})  \right). 
\end{eqnarray}
By the independence of $B^1, B^2$, the product $f^1_t(x_1)f^2_t(x_2)$ gives the approximation in \eqref{eq:app_q}.  
\end{remark}
\begin{figure}[h!tbp]
\centering
\begin{subfigure}{0.45\textwidth}
\includegraphics[width=\textwidth]{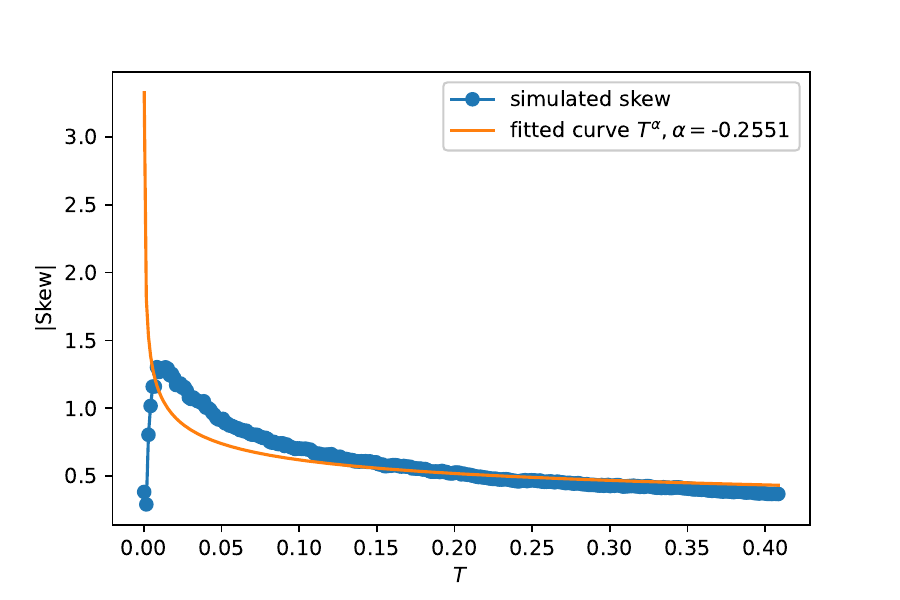}
\caption{$s^1_0 = 100,  s^2_0 = 96$.}
\label{fig:GBM_quasi}
\end{subfigure}
\qquad
\begin{subfigure}{0.45\textwidth}
\includegraphics[width=\textwidth]{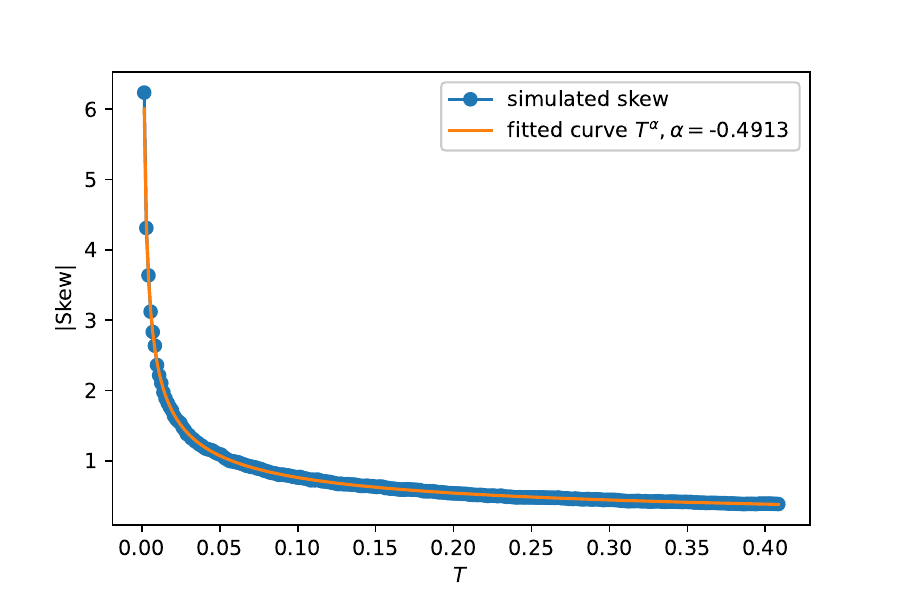} 
\caption{$s^1_0 = s^2_0 = 100$.}
\label{fig:GBM_blowup}
\end{subfigure}
\begin{subfigure}{0.45\textwidth}
\includegraphics[width=\textwidth]{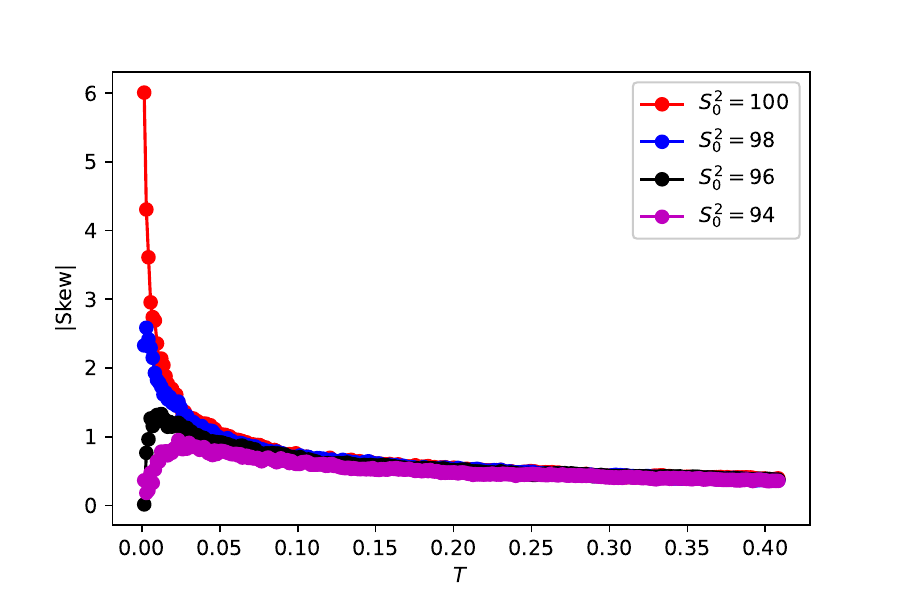}
\caption{$s^1_0 = 100, s^2_0 \in \{100,98,96,94\}$.}
\label{fig:GBM_quasi_many}
\end{subfigure}
\caption{The ATM skews at different maturities for the case with two geometric Brownian motions with parameters $dt = 0.05 \times 1/365, \sigma^1 = 0.2, \sigma^2 = 0.6, \overline{n} = 1, w_1 = 1, w_2=0$ and $50000$ Monte Carlo simulations. The Euler scheme is used.}
\label{fig:GBM}
\end{figure}

Let $\overline{n} \in \{1,2\}$ and $\omega_1, \omega_2$ be positive constants. In this example, the market index is given by $I_t = \omega_1S^{(1)}_t +  \omega_2S^{(2)}_t$. We distinguish two following cases. 
\begin{itemize}
\item[(i)] The case  $s^1_ 0 > s^2_0$. It can be seen that $m^k_1 = m^{k,j}_2= m^{kj}_3 = m^{kj\ell}_4 =  0, j =1, 2$ and then by Proposition \ref{pro:future1},
$$F_{0,T} = I_0 +  O\left( T \right) .$$
By Corollary \ref{cor_05}, the ATM skew does not blow up in this case. 
\item[(ii)] The case $s^1_0 = s^2_0$.  It can be checked that 
\begin{eqnarray*}
m^k_5 &=&  v^k_0(0) \left(  \int_{-\infty}^{\infty}  \int_{-\infty}^{\frac{\sigma^1}{\sigma^2}x_{1}} { w_ks^k_0x_k  \phi_{\mathbf{0},\mathbf I_2}\left( x\right) dx_2 dx_1}  + \int_{-\infty}^{\infty}  \int_{-\infty}^{\frac{\sigma^2}{\sigma^1}x_{2}} { w_ks^k_0x_k  \phi_{\mathbf{0},\mathbf I_2}\left( x\right) dx_1 dx_2}\right) ,\\
m^{kj}_6 &=& m^{kj}_7 = m^{kj\ell}_8 = 0,  k,j, \ell = 1,2. 
\end{eqnarray*} 
Proposition \ref{pro:future2} yields
$$F_{0,T} = I_0 + \sum_{j = 1}^n m^j_5 \sqrt{T}+  O\left( T\right). $$
By Corollary \ref{cor_05}, when $\sum_{k=1}^{\overline{n}}m^k_5 \ne 0$, the ATM skew blows up at the rate $T^{-1/2}$ in this case. 
\end{itemize}
Finally, from Lemma \ref{lem:lim} we have for any fixed $T >0$, 
\begin{equation*}
\lim_{s^2_0 \to s^1_0} \left. \frac{\partial \sigma^{IV}}{\partial k}(s^1_0,s^2_0)\right|_{k=0} = \left. \frac{\partial \sigma^{IV}}{\partial k}(s^1_0,s^1_0)\right|_{k=0}.
\end{equation*}
Therefore, the model exhibits the quasi-blow-up phenomena. 

We illustrate the theoretical findings
from the case (i) and case (ii) above
in Figure \ref{fig:GBM_quasi}- Figure \ref{fig:GBM_quasi_many}.
Concretely, in these figures, we choose $dt = 0.05 \times 1/365, \sigma^1 = 0.2, \sigma^2 = 0.6, \overline{n}=1, w_1 = 1, w_2=0$ 
and use Euler scheme to obtain the ATM implied volatility skew
by averaging  $50000$ Monte Carlo simulations. We then plot, in Figure \ref{fig:GBM_quasi}- Figure \ref{fig:GBM_blowup}, the absolute of the ATM implied volatility skew (colored dots) as the function of maturity $T$
and the corresponding fitted curves. 
Figure \ref{fig:GBM_quasi} illustrates the case (i): $s^1_ 0 > s^2_0$ discussed above.
We choose the initial prices to be $s^1_0 = 100,  s^2_0 = 96$.
It is clear that the ATM skew does not blow up as the blue star dots curve down toward the origin as the maturity diminishes.
On the other hand, Figure \ref{fig:GBM_blowup} clearly shows the
blow up of the skew at the rate $T^{-1/2}$ as the two stocks start at the same value $s_0^1=s_0^2=100$.
This confirms the theoretical finding in Corollary \ref{cor_05}. 
Finally, in Figure \ref{fig:GBM_quasi_many}, we plot the ATM skews and the corresponding fitted curves
from different starting pairs of initial stock values $(s^1_0,s^2_0)\in\{(100,94),(100,96),(100,98),(100,100)\}$.
It is evident to see that the skew exhibits the quasi-blow-up phenomenon as $s^2_0$ converges to $s_0^1$.


\subsection{Modified fractional Stein–Stein models}
We consider the following model
\begin{eqnarray*}
dS^j_t &=& S^j_t  \sigma^j_t(\rho^jdB^j_t + \sqrt{1-(\rho^j)^2}dW^j_t),\\
\sigma^j_t &=& \frac{\sigma^j_0}{c^j(t)} \left( c^j_0 + B^{H^j}_te^{-(B^{H^j}_t)^2/2} \right) ,\\
B^{H^j}_t &=& \frac{1}{\Gamma(H + 1/2)} \int_0^t(t-s)^{H^j-1/2}dB^j_s,
\end{eqnarray*}
where $W^j, B^j$ are independent Brownian motions, $\rho^1, \rho^2 \in [-1,1]$ for $j= 1,2.$ The constants $c^j_0, j = 1,2$ are big enough so that $\sigma^j_t$ are away from zero uniformly and
\begin{equation}\label{eq:c_t}
(c^j(t))^2 = (c^j_0)^2 + E[ (B^{H^j}_t)^2e^{-(B^{H^j}_t)^2}] = (c^j_0)^2 + \frac{1}{\sqrt{2}} \frac{t^{2H^j}}{(1/2+t^{2H^j})^{3/2}}.
\end{equation}
When $t$ is small enough, the processes $\sigma^j_t, j = 1,2$ behave as a fractional Brownian motion and  the model acts as the fractional Stein-Stein ones, see  \citet{abi2022characteristic},
\citet{gulisashvili2019extreme}, \citet{forde2017asymptotics}. Straightforward computations yield
\begin{eqnarray*}
v^j_0(t) &=&  \frac{(\sigma^j_0)^2}{(c^j(t))^2}\left( (c^j_0)^2 + E[ (B^{H^j}_t)^2e^{-(B^{H^j}_t)^2}]\right)  =(\sigma^j_0)^2,\\
V^j_0(t) &=& \sigma^j_0 \sqrt{t},\\
X^j_t &=& -\frac{1}{2 V^j_0(t)} \int_0^t (\sigma^j_u)^2du + \frac{1}{V^j_0(t)} \int_0^t \sigma^j_u (\rho^jdB^j_u + \sqrt{1-(\rho^1)^2}dW^j_u).
\end{eqnarray*}
\begin{remark}
Here, we normalize the volatility processes $\sigma^j_t$ by $c^j(t)$ so that $V^j_0(t)$ is of order $\sqrt{t}$  for simpler computations later on. Without such normalization and assume $\sigma^j_0 = 1$, i.e., 
$$\sigma^j_t = \left( c^j_0 + B^{H^j}_te^{-(B^{H^j}_t)^2/2} \right)$$
we obtain from \eqref{eq:c_t} that 
\begin{eqnarray*}
v^j_0(t) &=& (c^j_0)^2 + E[ (B^{H^j}_t)^2e^{-(B^{H^j}_t)^2}] = (c^j_0)^2 + \frac{1}{\sqrt{2}} \frac{t^{2H^j}}{(1/2+t^{2H^j})^{3/2}}
\end{eqnarray*}
and $V^j_0(t) = c^j_0\sqrt{t} + O(T^{1/2 + 2H^j})$. In this case, $\zeta^j = 2H^j$.
\end{remark} We follow \citet{euch2019short} to find the density expansion for $X^j_t, j = 1,2$. Fix $\theta > 0$ and define
\begin{eqnarray*}
\tau_j(s) &=& \frac{1}{V^j_0(\theta)^2} \int_0^s v^j_0(t)dt = \frac{s}{\theta}, \ s \le \theta,\\
\widehat{W}^j_t &=& \frac{1}{V^j_0(\theta)} \int_0^{\tau_j^{-1}(t)}\sqrt{v^j_0(s)} dW^j_s,\\
\widehat{B}^j_t &=& \frac{1}{V^j_0(\theta)} \int_0^{\tau_j^{-1}(t)}\sqrt{v^j_0(s)} dB^j_s.
\end{eqnarray*}
Then $\widehat{W}^j, \widehat{B}^j, j = 1,2$ are independent Brownian motions and note that $\tau_j^{-1}(t) = \theta t$. For any square integrable function $f$, we have
\begin{eqnarray}
\int_0^a f(s)dW^j_s &=& V^j_0(\theta) \int_0^{\tau_j(a)} \frac{f(\tau_j^{-1}(t))}{\sqrt{v_0(\tau_j^{-1}(t))}} d\widehat{W}^j_t, \label{eq:timechange1} \\ 
\int_0^a f(s)dB^j_s &=& V^j_0(\theta) \int_0^{\tau_j(a)} \frac{f(\tau_j^{-1}(t))}{\sqrt{v_0(\tau_j^{-1}(t))}} d\widehat{B}^j_t, \label{eq:timechange3}\\
\int_0^a f^2(s)ds &=& (V^j_0(\theta))^2\int_0^{\tau_j(a)} \frac{f^2(\tau_j^{-1}(t))}{v_0(\tau_j^{-1}(t))} dt. \label{eq:timechange2}
\end{eqnarray}
\begin{figure}[h!tbp]
\centering	
\begin{subfigure}{0.45\textwidth}
\includegraphics[width=\textwidth]{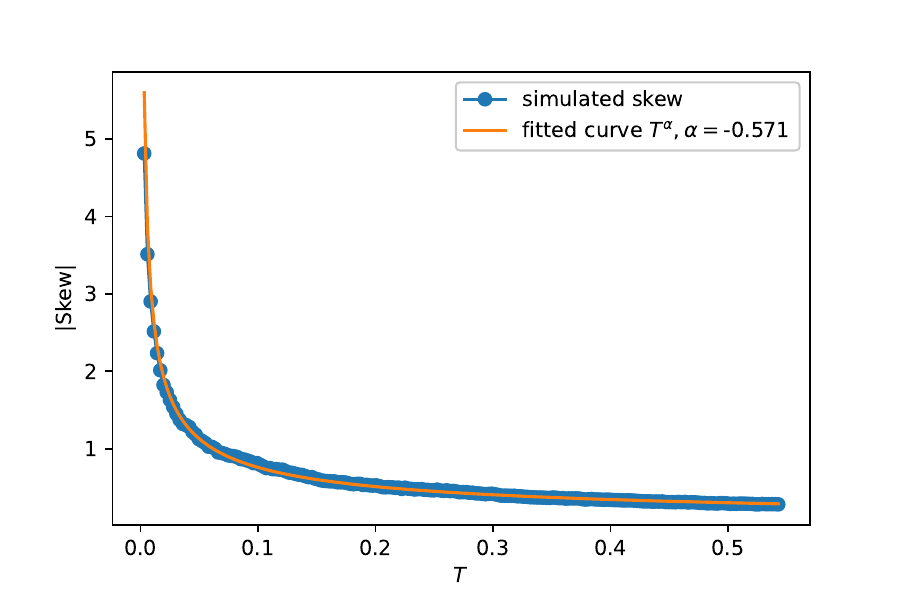}
\caption{$H^1 = 0.6, H^2 = 0.7,   s^1_0 =100, s^2_0 = 100$.}
\label{fig:fSS_blowup1_new7}
\end{subfigure} \qquad 
\begin{subfigure}{0.45\textwidth}	\includegraphics[width=\textwidth]{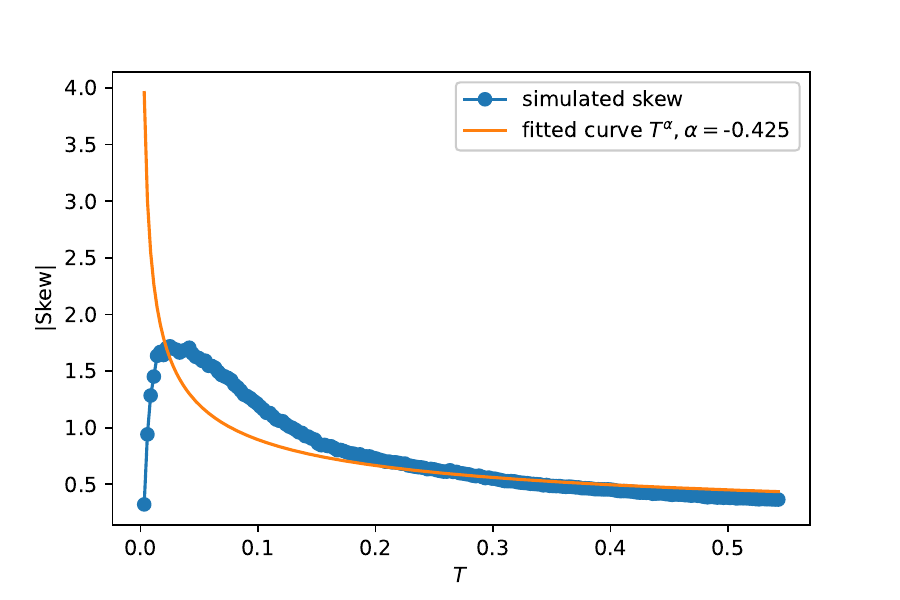}
\caption{$H^1 = 0.6, H^2 = 0.7,  s^1_0 =100, s^2_0 = 97$.}
\label{fig:fSS_blowup1_new8}
\end{subfigure}

\begin{subfigure}{0.45\textwidth}
\includegraphics[width=\textwidth]{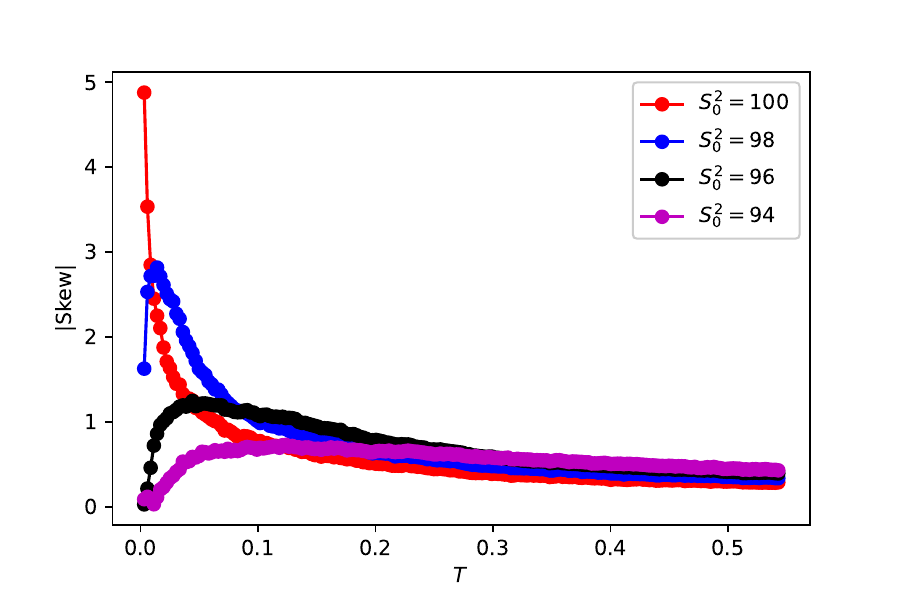}
\caption{$H^1 = 0.6, H^2 = 0.7$.}
\label{fig:fSS_blowup_many2}
\end{subfigure}
\caption{The ATM skews at different maturities for the case with two modified fractional Stein-Stein stocks with parameters $dt = 0.1 \times 1/365, \sigma^1 = 0.2, \sigma^2 = 0.6, \rho^1 = -0.5, \rho^2 = -0.5,  w_1 = 1, w_2 = 0$. In Figures \ref{fig:fSS_blowup1_new7}, \ref{fig:fSS_blowup1_new8}, we use $30000$ Monte Carlo simulations and in Figure \ref{fig:fSS_blowup_many2} we use $15000$ Monte Carlo simulations.}
\label{fig:fss1}
\end{figure}
From \eqref{eq:c_t}, we denote $$h^j_{\theta}(t) = \frac{1}{c^j(\tau^{-1}(t))} = \frac{1}{\sqrt{(c^j_0)^2 + \frac{1}{\sqrt{2}} \frac{(\theta t)^{2H^j}}{(1/2+(\theta t)^{2H^j})^{3/2}}}}.$$
\begin{lemma}\label{lem:fSS}
Assumption \ref{assum:main} is satisfied with
\begin{eqnarray}
M^{(0),j}_{\theta} &=&   \rho^j\widehat{B}^j_1 + \sqrt{1-(\rho^j)^2}\widehat{W}^j_1, \nonumber \\
M^{(1),j}_{\theta} &=& \int_0^1  h^j_{\theta}(t)F^{j,t}_t    (\rho^jd\widehat{B}^j_t + \sqrt{1-(\rho^j)^2}d\widehat{W}^j_t ), \nonumber \\
M^{(2),j}_{\theta} &=&  \int_0^1 \left( \frac{h^j_{\theta}(t)c^j_0-1}{\theta^{2H^j}}  \right) (\rho^jd\widehat{B}^j_t + \sqrt{1-(\rho^j)^2}d\widehat{W}^j_t ), \nonumber  \\ 
M^{(3),j}_{\theta} &=& \frac{2}{c^j_0} \int_0^1 F^{j,t}_t dt . \label{fSS_M}
\end{eqnarray}
\end{lemma}
\begin{proof}
It could be checked that the condition \eqref{eq:vola} is satisfied. From \eqref{eq:timechange1}, \eqref{eq:timechange3} we get that
\begin{eqnarray*}
\frac{M^j_{\theta}}{V^j_0(\theta)} &=&  \int_0^1  \frac{c^j_0 + \theta^{H^j}F^{i,t}_te^{-\theta^{2H^j}(F^{i,t}_t)^2/2}}{c^j(\tau^{-1}_j(t))} (\rho^jd\widehat{B}^j_t + \sqrt{1-(\rho^j)^2}d\widehat{W}^j_t )
\end{eqnarray*}
where 
\begin{eqnarray*}
F^{j,t}_u:=\frac{1}{\Gamma(H^j+1/2)} \frac{V^j_0(\theta)}{\theta^{H^j}}  \int_0^{u} \frac{\left( \tau_j^{-1}(t) - \tau_j^{-1}(s) \right)^{H^j-1/2}}{\sqrt{v^j_0(\tau_j^{-1}(s))}} d\widehat{B}^j_s, \  u \in [0,t].
\end{eqnarray*}
The function $f(x) = xe^{-x^2}/2$ is in the Schwartz space, and Taylor's theorem implies
\begin{equation}\label{eq:expansionxe}
xe^{-x^2/2} =  x + \frac{f^{(3)}(\xi)}{6} x^3,
\end{equation}
for some $\xi$ is between $0$ and $x$. We write
\begin{eqnarray*}
\frac{M^j_{\theta}}{V^j_0(\theta)} &=&  \int_0^1 h^j_{\theta}(t) \left(c^j_0 +  \theta^{H^j}F^{j,t}_t + \frac{f^{(3)}(\xi_t)}{6}\theta^{3H^j}(F^{j,t}_t)^3 \right)  (\rho^jd\widehat{B}^j_t + \sqrt{1-(\rho^j)^2}d\widehat{W}^j_t )
\end{eqnarray*}
for some $\xi_t$ is between $0$ and $\theta^{H^j}F^{j,t}_t $. Define $M^{(0),j}, M^{(1),j}, M^{(2),j}$ as in \eqref{fSS_M}, and note that $M^{(0),j}$ is a Gaussian random variable and $h^j_{\theta}(t)c^j_0 - 1 = O(\theta^{2H^j})$. We need to prove that
\begin{equation}\label{eq:2}
\left\| \int_0^1 h^j_{\theta}(t) \frac{f^{(3)}(\xi_t)}{6}\theta^{3H^j}(F^{j,t}_t)^3   (\rho^jd\widehat{B}^j_t + \sqrt{1-(\rho^j)^2}d\widehat{W}^j_t ) \right\|_2 = O(\theta^{3H^{j}}) 
\end{equation}
and then the condition \eqref{assum:M012} holds. Indeed, Since $B^{H^j}_t \sim t^{H^j} N$ where $N \sim N(0,1)$, we have $E[|B^{H^j}_t|^{p}] = t^{pH^j}E[|N|^{p}]$ for any $p>1$ and thus 
\begin{equation}\label{eq:moment}
E\left[|F^{j,t}_t|^{p}\right] = E[|\theta^{-H^j}B^{H^j}_{\tau_j^{-1}(t)}|^{p}] =  E[|N|^{p}] \frac{|\tau_j^{-1}(t)|^{pH^j}}{\theta^{pH^j}} = (p-1)!! t^{pH^j}	
\end{equation}
because $\tau^{-1}(t) \le \tau^{-1}(1) = \theta$. We estimate by the Burkholder-Davis-Gundy  inequality \begin{eqnarray*}E\left[  \left( \int_0^1 h^j_{\theta}(t) \frac{f^{(3)}(\xi_t)}{6}(F^{j,t}_t)^3   (\rho^jd\widehat{B}^j_t + \sqrt{1-(\rho^j)^2}d\widehat{W}^j_t )\right)^2\right]   &\le& C  \int_0^1E[(F^{j,t}_t)^6] dt \\
&\le& C5!!,
\end{eqnarray*}
for some constant $C>0$, noting that $h^j_{\theta}(t)$ and $f^{(3)}(x)$ are bounded from above. Hence, \eqref{eq:2} holds. Next, using \eqref{eq:timechange2}, we compute  
\begin{eqnarray*}
\frac{\left\langle M^j_{\theta}\right\rangle }{(V^j_0(\theta))^2} &=&  \int_0^1  \frac{(c^j_0 + \theta^{H^j}F^{j,t}_te^{-\theta^{2H^j}(F^{j,t}_t)^2/2})^2 }{(c^j(\tau_j^{-1}(t)))^2} dt \\
&=& 1  + \int_0^1\left( (h^j_{\theta}(t)c^j_0)^2 -1 \right)  dt  + \int_0^1  (h^j_{\theta}(t))^22c^j_0  \theta^{H^j} F^{j,t}_te^{-\theta^{2H^j}(F^{j,t}_t)^2/2} dt \\
&+& \int_0^1 (h^j_{\theta}(t))^2 \theta^{2H^j}(F^{j,t}_t)^2e^{-\theta^{2H^j}(F^{j,t}_t)^2} dt.
\end{eqnarray*}
Using the expansion \eqref{eq:expansionxe} again, we write
\begin{eqnarray*}
\int_0^1  (h^j_{\theta}(t))^22c^j_0  \theta^{H^j} F^{j,t}_te^{-\theta^{2H^j}(F^{j,t}_t)^2/2} dt &=& \int_0^1  (h^j_{\theta}(t))^22c^j_0  \theta^{H^j} F^{j,t}_tdt + \int_0^1  (h^j_{\theta}(t))^22c^j_0  \frac{f^{(3)}(\xi_t)}{6}\theta^{3H^j}(F^{j,t}_t)^3dt \\
&=& \int_0^1  \frac{2}{c^j_0}  \theta^{H^j} F^{j,t}_tdt  +  \int_0^1  \left( (h^j_{\theta}(t))^22c^j_0 -  \frac{2}{c^j_0}  \right)   \theta^{H^j} F^{j,t}_tdt \\
&+& \int_0^1  (h^j_{\theta}(t))^22c^j_0  \frac{f^{(3)}(\xi_t)}{6}\theta^{3H^j}(F^{j,t}_t)^3dt.  
\end{eqnarray*}
From these expressions, we get $M^{(3),j}$ and \eqref{assum:M3} holds, noting that 
$$(h^j_{\theta}(t)c^j_0)^2 -1 = O(\theta^{2H^j}), \qquad (h^j_{\theta}(t))^22c^j_0 - \frac{2}{c^j_0} = O(\theta^{2H^j}).$$
The condition \eqref{assum:moment} is satisfied from \eqref{eq:moment}. The computations of $a^{(k),j}_{\theta}(\bx)$, $ b^j_{\theta}(\bx)$, $c^j_{\theta}(\bx)$, $d^{(1),j,k}_{\theta}(\bx)$, $e^{(1),j,k(\bx)}_{\theta}$ are almost similar to the computations given in Lemmas 5.2, 5.3, 5.4, 5.5 of \citet{euch2019short} and hence omitted. Finally, Assumption \ref{assumption:Gaussian_bound} is satisfied because $\sigma^j_t, t = 1,2$ are bounded from above. 
\end{proof}
We illustrate the theoretical findings
above in Figures \ref{fig:fss1}, \ref{fig:fss2}.
We use {the Cholesky method to simulate fractional Brownian motions and the Euler scheme for stock prices}. We then plot the absolute of the ATM implied volatility skew (colored dots) as the function of maturity $T$
and the corresponding fitted curves. In Figure \ref{fig:fSS_blowup1_new7}, all the two Hurst parameters are larger than $1/2$. The ATM skew blows up at the rate $T^{-1/2}$ when the initial stock values are the same. And if the initial stock values are close, the ATM skew exhibits quasi-blow-up, see Figure \ref{fig:fSS_blowup1_new8}. These results are  explained by Corollary \ref{cor_05}.   On the hand, in Figure \ref{fig:fSS_blowup_many2}, we plot the ATM skews from different starting pairs of initial stock values $(s^1_0,s^2_0)\in\{(100,94),(100,96),(100,98),(100,100)\}$.
Again, the  ATM skew exhibits the quasi-blow-up phenomenon as $s^2_0$ converges to $s_0^1$. 

Next, we assume that at least one Hurst parameter is smaller than $1/2$. As we see in Figures \ref{fig:fSS_blowup_same1}, \ref{fig:fSS_blowup1_new4},  when the two stocks have the same initial value, the ATM skew blows up at the rate $T^{-1/2}$. This is consistent with Corollary \ref{cor:<0.5}. When the initial values are different, see Figure \ref{fig:fSS_blowup_diff2}, the ATM skew blows up at the rate $T^{-0.35}$ approximately. This is because $w_1 = 1, w_2 = 0$, and for small $T$, the index is well approximated by the first stock with $H^1 = 0.2$. The blow up rate in this case is similar to that of the one dimensional fractional Stein-Stein model, i.e., $T^{H^1-1/2}$. In Figure \ref{fig:fSS_blowup_diff3}, the two weight parameters are non-zero, the ATM skew blows up at the rate $T^{-0.319} \sim T^{H^2 - 1/2}$, as predicted by Corollary \ref{cor:<0.5}. 
\begin{figure}[h!tbp]
\centering
\begin{subfigure}{0.45\textwidth}
\includegraphics[width=\textwidth]{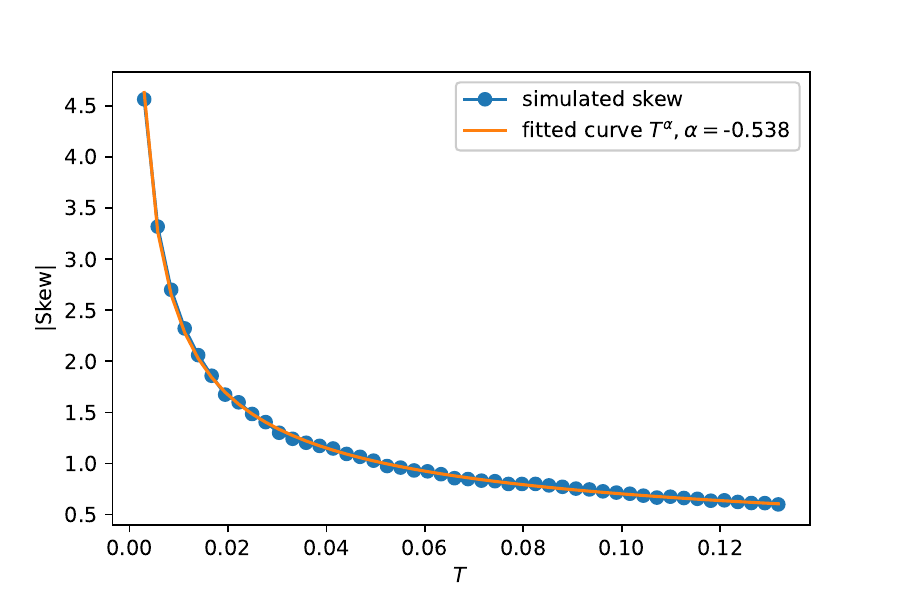}
\caption{$H^1 = 0.3, H^2 = 0.4,   s^1_0 =100, s^2_0 = 100, w_1 = 1, w_2 = 0$.}
\label{fig:fSS_blowup_same1}
\end{subfigure} \qquad
\begin{subfigure}{0.45\textwidth}
\includegraphics[width=\textwidth]{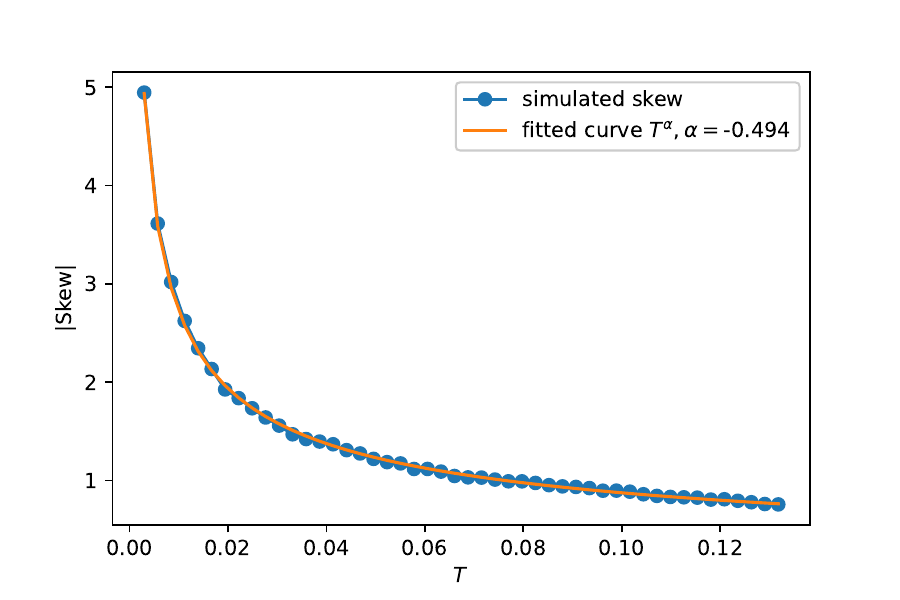}
\caption{$H^1 = 0.2, H^2 = 0.7,  s^1_0 =100, s^2_0 = 100, w_1 = 1, w_2 = 0$.}
\label{fig:fSS_blowup1_new4}
\end{subfigure}
\begin{subfigure}{0.45\textwidth}
\includegraphics[width=\textwidth]{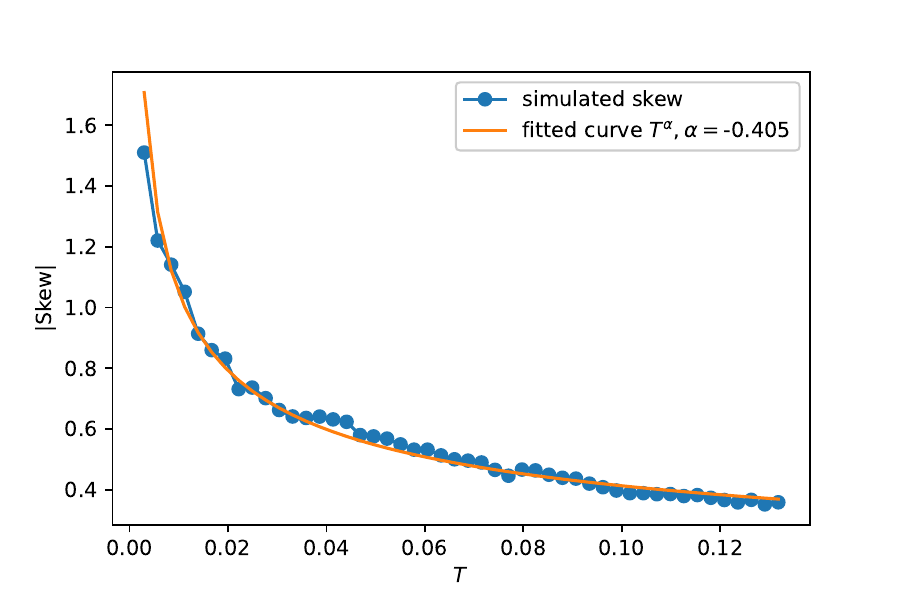}
\caption{$H^1 = 0.2, H^2 = 0.3,  s^1_0 =100, s^2_0 = 90, w_1 = 1, w_2 = 0$.}
\label{fig:fSS_blowup_diff2}
\end{subfigure}\qquad
\begin{subfigure}{0.45\textwidth}
\includegraphics[width=\textwidth]{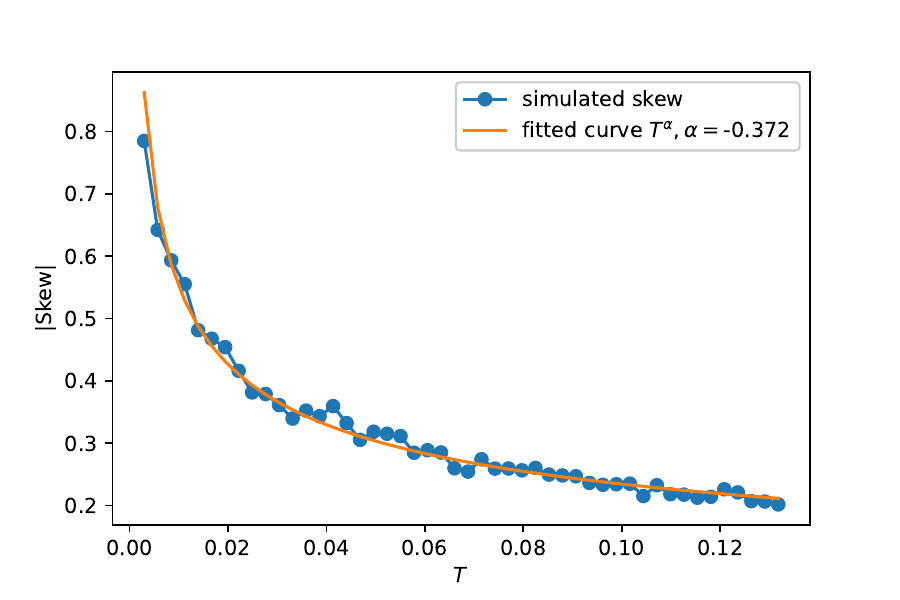}
\caption{$H^1 = 0.7, H^2 = 0.2,  s^1_0 =100, s^2_0 = 90, w_1 = 0.7, w_2 = 0.3$.}
\label{fig:fSS_blowup_diff3}
\end{subfigure}
\caption{The ATM skews at different maturities for the case with two modified fractional Stein-Stein stocks with parameters $dt = 0.1 \times 1/365, \sigma^1 = 0.2, \sigma^2 = 0.6, \rho^1 = -0.5, \rho^2 = -0.5,$ and $30000$ Monte Carlo simulations.}
\label{fig:fss2}
\end{figure}
\subsection{Fractional Bergomi models}
We now consider the fractional Bergomi model. In particular, we assume that
\begin{eqnarray*}
\sigma^1_t &=& \sigma^1_0 \exp \left\lbrace \eta_1 \sqrt{2H^1} \int_0^t (t-s)^{H^1 - 1/2} dB^1_s - \frac{\eta_1^2}{2} t^{2H^1}
\right\rbrace,\\
\sigma^2_t &=& \sigma^2_0 \exp \left\lbrace \eta_2 \sqrt{2H^2} \int_0^t (t-s)^{H^2- 1/2} dB^2_s - \frac{\eta_2^2}{2} t^{2H^2} 
\right\rbrace,
\end{eqnarray*}
and the log-price processes satisfy
\begin{eqnarray*}
dZ^1_t &=& -\frac{1}{2} \sigma^1_t dt + \sqrt{\sigma^1_t} (\rho_1 dB^1_t + \sqrt{1-\rho_1^2}dW^{1}_t),\\
dZ^2_t &=& -\frac{1}{2} \sigma^2_t dt + \sqrt{\sigma^2_t} (\rho_2 dB^2_t + \sqrt{1-\rho_2^2}dW^{2}_t). 
\end{eqnarray*}
where $B^j,W^{j}, j = 1,2$ are independent Brownian motions. We can follow \citet{euch2019short} for the computation of $q_T(\bx)$. The numerical results are reported in Figure \ref{fig:fBM}, in which we still observe the blow up and quasi-blow up phenomena.
\begin{figure}[h!tbp]
\centering
\begin{subfigure}{0.45\textwidth}
\includegraphics[width=\textwidth]{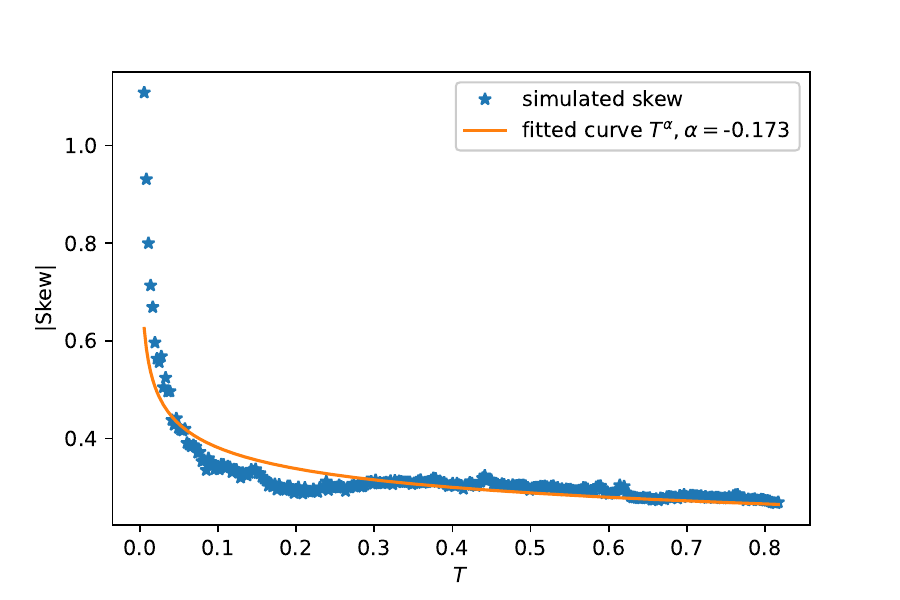} 
\caption{$s^1_0 = s^2_0 = 100$.}
\label{fig:fBM_blowup}
\end{subfigure}
\begin{subfigure}{0.45\textwidth}
\includegraphics[width=\textwidth]{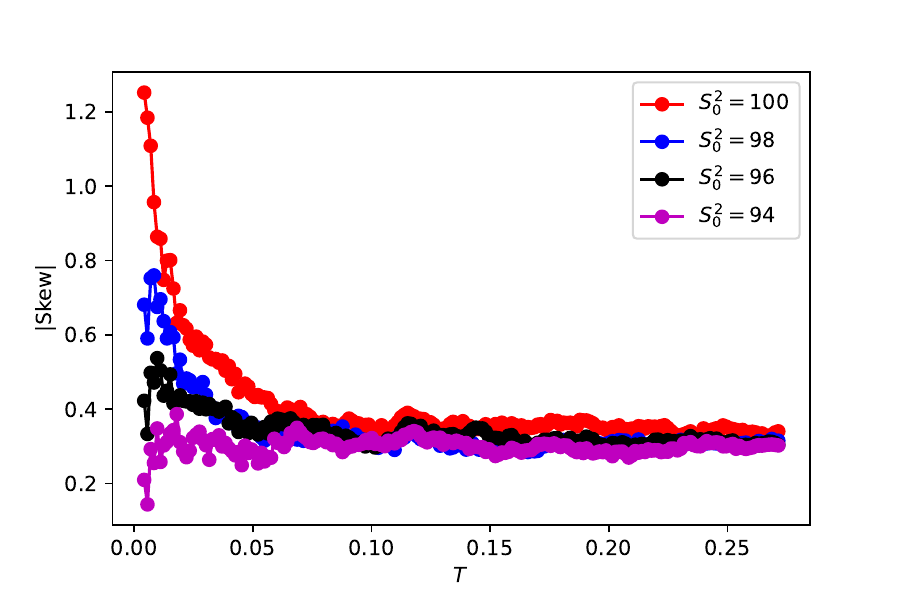}
\caption{$s^1_0 = 100,  s^2_0 \in \{100,98,96,94\}$.}
\label{fig:fBM_quasi_many}
\end{subfigure}
\caption{The ATM skews at different maturities for the case with two fractional Bergomi models with parameters $dt = 0.1/365, H^1 = 0.7, H^2 = 0.6, \eta_1 = \eta_2 = 1.9^2, \rho_1 = \rho_2 =0,  w_1 = 1, w_2  = 0$ and $30000$ Monte Carlo  simulations. The Hybrid scheme of \citet{bennedsen2017hybrid} is used. The model exhibits the quasi-blow-up phenomena.}	\label{fig:fBM}
\end{figure}
However, Theorem 6.7 of \citet{gulisashvili2020gaussian} implies that for any $T > 0, p > 0$
\begin{eqnarray*}
E\left[ \exp\left\lbrace p\int_0^T (\sigma^j_t)^2dt\right\rbrace \right]  &\ge& E\left[ \exp\left\lbrace p(\sigma^1_0)^2  e^{- \eta_1^2 T^{2H_1} }  \int_0^T \exp \left\lbrace 2\eta_1 \sqrt{2H^1} \int_0^t (t-s)^{H_1 - 1/2} dB^1_s   \right\rbrace   dt \right\rbrace \right] \\
&=& \infty,  
\end{eqnarray*}
and thus Assumption \ref{assumption:Gaussian_bound} is not satisfied. We also refer to \citet{gassiat} for similar results.  Therefore, weakening Assumption \ref{assumption:Gaussian_bound} or developing different methods to study fractional Bergomi models would be interesting research questions. 
\section{Conclusion}
We have introduced a new market model that incorporates market indexes. This model involves ranking stock prices based on their capitalization and subsequently constructing the market indexes from the top-ranked stocks. Even in straightforward settings where stock prices follow geometric Brownian motion dynamics, the ranking mechanism has the capability to reproduce the observed term structure of ATM implied volatility skew for equity indexes. Additionally, we have developed models that resolve two perplexing empirical observations in equity markets: the persistent nature of volatilities and the power-law behavior of ATM skews. This is accomplished by incorporating fractional Brownian motions with Hurst exponents larger than 0.5 for volatilities and by implementing the ranking procedure. Our framework introduces a new phenomenon termed ``quasi-blow-up" and provides a comprehensive explanation for it.
Extensive numerical examples validate our theoretical findings.

\section{Proofs}\label{sec:proofs}
\subsection{Proof of Theorem \ref{thm:density_expansion}}\label{proof:density}
\begin{lemma}\label{lemma_any_power}
There exists a density of $\bX_t$ and for any $j \in \mathbb{N}$
$$\sup_{t \in (0,1)} \int|\bu|^j |E[e^{i\bu \cdot \bX_t}]|d\bu < \infty.$$
\end{lemma}
\begin{proof}
The proof is inspired by that of Lemma 3.4 of \citet{euch2019short}. 
\end{proof}
For each $j = 1,...,n,$ define $\widetilde{\bX}_t := (\widetilde{X}^j_t)_{1\le j \le n}$ where 
\begin{equation}
\widetilde{X}^j_t = M^{j,(0)}_t + t^{H^j} M^{j,(1)}_t + t^{2H^j}M^{j,(2)}_t - \frac{V^j_0(t)}{2} \left( 1 + t^{H^j}M^{j,(3)}_t \right).
\end{equation}
First, we will prove that for $\bu=(u_1,...,u_n) \in \mathbb{R}^n$,
\begin{equation}\label{eq:compare_char}
\sup_{|\bu| \le t^{-\varepsilon}}\left| E\left[ e^{i \bu \cdot \bX_t} \right] - E\left[  e^{i\bu \cdot \widetilde{\bX}_t} \right] \right| = \sum_{j=1}^n o(t^{H^1 + \min\{H^j,1/2\} + 2\varepsilon}). 
\end{equation} 
Decompose
\begin{eqnarray}
e^{i\bu \cdot \bX_t}  - e^{i \bu \cdot \widetilde{\bX}_t} &=&  \prod_{j=1}^{n}e^{iu_jX^j_t}  - e^{iu_1\widetilde{X}^1_t} \prod_{j=2}^{n}e^{iu_jX^j_t} 
+   e^{iu_1\widetilde{X}^1_t} \prod_{j=2}^{n}e^{iu_jX^j_t} -   e^{iu_1\widetilde{X}^1_t} e^{iu_2\widetilde{X}^2_t}  \prod_{j=3}^{n}e^{iu_jX^j_t} \nonumber \\
&+& \cdots +  \prod_{j=1}^{n-1}e^{iu_j\widetilde{X}^j_t}  e^{iu_nX^n_t} - \prod_{j=1}^{n}e^{iu_j\widetilde{X}^j_t}. \label{eq:decomp} 
\end{eqnarray}
We compute the first term in RHS of \eqref{eq:decomp}, other terms are treated similarly. Using $|e^{ix} - 1| \le |x|$, H\"older's inequality and the fact that $X^1_t, \widetilde{X}^1_t$ have moments of any order by \eqref{assum:inte}, \eqref{assum:moment}, we estimate
\begin{eqnarray*}
\left| E\left[ \prod_{j=1}^{n}e^{iu_jX^j_t}  -  e^{iu_1\widetilde{X}^1_t} \prod_{j=2}^{n}e^{iu_jX^j_t} \right] \right| 
&\leq& \left|  E\left[  e^{iu_1\widetilde{X}^1_t}(e^{iu_1(X^1_t-\widetilde{X}^1_t)}-1) \prod_{j=2}^{n}e^{iu_jX^j_t}  \right] \right| \\
&\leq&C(\varepsilon) t^{-\varepsilon}  \left\| X^1_t-\widetilde{X}^1_t \right\|_{1+ \varepsilon},
\end{eqnarray*}
for some $0 < \varepsilon$  small enough and $C(\varepsilon)> 0$.  Furthermore, we have $V^1_0(t) = O(t^{1/2})$ and then $\left\| X^1_t-\widetilde{X}^1_t \right\|_{1+ \varepsilon} = o(t^{H^1 + \min\{H^1,1/2\} + 2\varepsilon})$ by \eqref{assum:M012}, \eqref{assum:M3}. Therefore \eqref{eq:compare_char} follows. 

Secondly, we prove that for $\delta \in (0, \min_{1\le j\le n}\{(H^j-\varepsilon)/3, (1/2 - \varepsilon)/3\})$, 
\begin{eqnarray}
&&\sup_{|\bu| \le t^{-\delta}} \left| E\left[ \prod_{j=1}^n e^{iu_j\widetilde{X}^j_t} -  \prod_{j=1}^n  e^{iu_jM^{(0),j}_t} \left(  1 + iu_j (\widetilde{X}^j_t - M^{(0),j}_t) -\frac{u_j^2}{2} (\widetilde{X}^j_t - M^{(0),j}_t)^2 \right) \right]   \right| \nonumber \\
&&  \qquad= \sum_{j=1}^no(t^{\min\{2H^j,1\} + \varepsilon}). \label{eq:1term}
\end{eqnarray}
Using a similar decomposition as in \eqref{eq:decomp}, we need to estimate 
\begin{equation*}
\left| E\left[ \prod_{j=1}^{n}e^{iu_j\widetilde{X}^j_t}  -  e^{iu_1M^{(0),1}_t} \left(  1 + iu_1 (\widetilde{X}^1_t - M^{(0),1}_t) -\frac{u^2}{2} (\widetilde{X}^1_t - M^{(0),1}_t)^2 \right) \prod_{j=2}^{n}e^{iu_j\widetilde{X}^j_t} \right] \right|, 
\end{equation*}
and other terms follow in the same manner. By the inequality
$$\left| e^{ix} - 1 - ix + \frac{x^2}{2}\right|  \le \frac{|x|^3}{6}, \ \forall x \in \mathbb{R},$$
the quantity in \eqref{eq:1term} is estimated as follows
\begin{eqnarray*}
&&\left| E\left[ \left(  e^{iu_1\widetilde{X}^1_t}  - e^{iu_1M^{(0),1}_t} \left(  1 + iu_1 (\widetilde{X}^1_t - M^{(0),1}_t) -\frac{u_1^2}{2} (\widetilde{X}^1_t - M^{(0),1}_t)^2 \right)\right)  \prod_{j=2}^{n}e^{iu_j\widetilde{X}^j_t} \right] \right| \\
&=& \left| E\left[  e^{iu_1M^{(0),1}_t} \left(e^{iu_1(\widetilde{X}^1_t-M^{(0),1}_t)}  -  \left(  1 + iu_1 (\widetilde{X}^1_t - M^{(0),1}_t) -\frac{u_1^2}{2} (\widetilde{X}^1_t - M^{(0),1}_t)^2 \right)\right)  \prod_{j=2}^{n}e^{iu_j\widetilde{X}^j_t} \right] \right| \\
&=& o(t^{\min\{2H^1,1\} + \varepsilon}). 
\end{eqnarray*}
Therefore, \eqref{eq:1term} follows. From \eqref{eq:1term}, taking conditional expectation given $\bM^{(0)}_t = \bx$ gives
\begin{eqnarray}
&& \sup_{|\bu| \le t^{-\delta}} \left| E\left[ \prod_{j=1}^{n} e^{iu_j\widetilde{X}^j_t} - \prod_{j=1}^n e^{iu_jM^{(0),j}_t} \left( 1 + A^j(u_j,\bM^{(0)}_t) + B^j(u_j,\bM^{(0)}_t) \right) \right]   \right| \nonumber \\
&& \qquad = \sum_{j=1}^no(t^{\min\{2H^j,1\} + \varepsilon}), \label{eq:compare_char2}
\end{eqnarray}
where 
\begin{eqnarray}
A^j  &:=&	A^j(u_j,\bM^{(0)}_t) = iu_j\left( E\left[\widetilde{X}^j_t|\bM^{(0)}_t = \bx \right] - x_j\right)  \\
B^j &:=& B^j(u_j,\bM^{(0)}_t) \\
&=& -\frac{u^2_j}{2}\left( t^{2H^j} E\left[ \left| M^{(1),j}_t \right|^2 | \bM^{(0)}_t = \bx  \right] - V^j_0(t)  t^{H^j}  E\left[  M^{(1),j}_t | \bM^{(0)}_t = \bx  \right] + \frac{(V^j_0(t))^2}{4} \right) \nonumber.
\end{eqnarray}  
If we ignore the terms with order smaller than $t^{2H^j}$, we get
\begin{eqnarray*}
\prod_{j=1}^n e^{iu_jM^{(0),j}_t} \left(  1 + A^j + B^j \right)  &=& e^{i \bu \cdot \bM^{(0)}_t}\left( 1 + \sum_{j=1}^n A^j + \sum_{j=1}^n B^j + \sum_{1\le k,j \le n} A^kA^j \right). 
\end{eqnarray*} 
Next, using Lemma \ref{lemma:intbypart} we obtain
\begin{eqnarray*}
E\left[ e^{i \bu \cdot \bM^{(0)}_t}\left( \sum_{j=1}^n A^j  \right)\right] &=& \int_{\mathbb{R}^d} e^{i \bu \cdot \bx } \sum_{j=1}^niu_j  \left( E\left[\widetilde{X}^j_t|\bM^{(0)}_t = \bx \right] - x_j\right) \phi_{\mathbf{\mu},\Gamma}(\bx)d\bx\\
&=&  \int_{\mathbb{R}^d} e^{i \bu \cdot \bx} \left(  \sum_{j=1}^nit^{H^j} u_j E\left[M^{(1),1}_t|\bM^{(0)}_t = \bx \right]\right) \phi_{\mathbf{\mu},\Gamma}(\bx)d\bx \\
&+&\int_{\mathbb{R}^d} e^{i \bu \cdot \bx}  \left( \sum_{j=1}^n it^{2H^j} u_j E\left[M^{(2),j}_t|\bM^{(0)}_t = \bx \right] \right) \phi_{\mathbf{\mu},\Gamma}(\bx)d\bx\\
&-& \int_{\mathbb{R}^d} e^{i \bu \cdot \bx} \left( \sum_{j=1}^n i  \frac{V^j_0(t)}{2} u_j\right) \phi_{\mathbf{\mu},\Gamma}(\bx)d\bx\\
&-& \int_{\mathbb{R}^d} e^{i \bu \cdot \bx } \left(  \sum_{j=1}^n i \frac{V^j_0(t)t^{H^j}}{2}u_j   E\left[M^{(3),j}_t|\bM^{(0)}_t = \bx \right]\right) \phi_{\mathbf{\mu},\Gamma}(\bx)d\bx. \\
\end{eqnarray*}
Then
\begin{eqnarray*}
E\left[ e^{i \bu \cdot \bM^{(0)}_t}\left( \sum_{j=1}^n A^j  \right)\right] 	&=&    \int_{\mathbb{R}^d} e^{i \bu \cdot \bx} \sum_{j=1}^{n} t^{H^j} a^{(1),j}_t(\bx) d\bx + \int_{\mathbb{R}^d} e^{i \bu \cdot \bx} \sum_{j=1}^{n} t^{2H^j} a^{(2),j}_t(\bx) d\bx \\
&-&  \int_{\mathbb{R}^d} e^{i \bu \cdot \bx} \sum_{j=1}^{n}  V^j_0(t) \frac{\partial}{\partial x_j} \phi_{\mathbf{\mu},\Gamma}(\bx)d\bx - \int_{\mathbb{R}^d} e^{i \bu \cdot \bx} \sum_{j=1}^{n} \frac{V^j_0(t)t^{H^j}}{2} \cdot a^{(3),j}_t(\bx)d\bx. 
\end{eqnarray*}
Similarly
\begin{eqnarray*}
E\left[ e^{i \bu \cdot  \bM^{(0)}_t}\left( \sum_{j=1}^n B^j  \right)\right] 
&=&  \int_{\mathbb{R}^d} e^{i \bu \cdot \bx } \sum_{j=1}^{n}   \frac{t^{2H^j}}{2} c^j_t(\bx) d\bx
- \int_{\mathbb{R}^d} e^{i \bu \cdot \bx} \sum_{j=1}^{n} \frac{V^j_0(t) t^{H^j}}{2} b^j_t(\bx)d\bx\\
&+& \int_{\mathbb{R}^d} e^{i \bu \cdot \bx}  \sum_{j=1}^{n} \frac{(V^j_0(t))^2}{8} \frac{\partial^2}{\partial x^2_j} \phi_{\mathbf{\mu},\Gamma}) (\bx)d\bx. 
\end{eqnarray*}
In addition,
\begin{eqnarray*}
E\left[ e^{i \bu \cdot  \bM^{(0)}_t}\left( \sum_{1\le k,j  \le n}^n A^kA^j  \right)\right] &=& \int_{\mathbb{R}^d} e^{i \bu \cdot \bx}   \sum_{1\le k,j  \le n}^n  t^{H^k + H^j}   d^{(1),j,k}_t(\bx)d\bx\\
&-& \int_{\mathbb{R}^d} e^{i \bu \cdot  \bx}  \sum_{1 \le k,j\le n}t^{H^j} \frac{V^k_0(t)}{2}e^{(1),j,k}_t(\bx)d\bx\\
&+&  \int_{\mathbb{R}^d} e^{i \bu \cdot  \bx}  \sum_{1 \le k,j\le n} \frac{V^k_0(t)V^j_0(t)}{4} \frac{\partial^2}{\partial x_jx_k}\phi_{\mathbf{\mu},\Gamma}(\bx)d\bx .
\end{eqnarray*}
Therefore,
\begin{eqnarray*}
E\left[  \prod_{j=1}^n e^{iu_jM^{j,(0)}_t} \left( 1 + A^j(u_j,\bM^{(0)}_t) + B^j(u_j,\bM^{(0)}_t) \right) \right]     = \int_{\mathbb{R}^n} e^{i \bu \cdot \bx} q_t (\bx) d\bx,
\end{eqnarray*}
where $q$ is given in \eqref{eq:q}. 
By the Fourier identity, we get
\begin{eqnarray*}
\left( p_t(\bx) - q_t(\bx)\right) = \frac{1}{2\pi}\int_{\mathbb{R}^n} \int_{\mathbb{R}^n}e^{i\bu \cdot  \by}\left( p_t(\by) - q_t(\by)\right)d\by e^{-i\bu \cdot \bx}d\bu .
\end{eqnarray*}
The volume of a Euclidean ball of radius $R$ in $n$-dimensional is of order $R^n$. Choosing $\delta \in (0, \min_{1 \le j \le n}\{\varepsilon/(2n), (H^j - \varepsilon)/3, (1/2 - \varepsilon)/3\})$ yields
\begin{eqnarray*}
\int_{|\bu| \le t^{-\delta}} \left|  \int_{}e^{i\bu \cdot  \by}\left( p_t(\by) - q_t(\by)\right)d\by \right| d\bu = \sum_{j=1}^n o(t^{\min\{2H^j,1\}+\varepsilon/2}).
\end{eqnarray*}
Furthermore
\begin{eqnarray*}
\int_{|\bu| > t^{-\delta}} \left|  \int_{\mathbb{R}^n}e^{i\bu \cdot  \by}p_t(\by) d\by \right| d\bu \le t^{j\delta}  \int_{|\bu| > t^{-\delta}} |\bu|^j  \left|  E \left[ e^{i\bu \cdot  \bX_t}\right] \right| d\bu =  O(t^{j\delta}),
\end{eqnarray*}
and similarly
\begin{eqnarray*}
\int_{|\bu| > t^{-\delta}} \left|  \int_{\mathbb{R}^n}e^{i\bu \cdot  \by}q_t(\by) d\by \right| d\bu = O(t^{j\delta}),
\end{eqnarray*}
for any $j \in \mathbb{N}$ by Lemma \ref{lemma_any_power}.  The proof is complete. 
\subsection{Proof of Proposition \ref{pro:future1}}\label{proof:future1}
Recall that $\Pi_n$ contains all permutations of $\{1,2,...,n\}$ and 
$$A^{\psi_n}_T = \{ \omega: S^{\psi_n(1)}_T \ge S^{\psi_n(2)}_T \ge \cdots \ge S^{\psi_n(n)}_T \}.$$  
By definition, the price of index future at time $0$  becomes 
\begin{eqnarray}\label{eq:future_price}
F_{0,T} &=& E[I_T|\mathcal{F}_0] = \sum_{\psi_n\in\Pi_n}E\left[ I_T1_{A^{\psi_n}_T}\right]. \label{eq:future_decomp}
\end{eqnarray}
In this case, the event $A^{(1,2,...,n)}_T$ is the largest one among all permutations as $T$ tends to $0$ and the quantity $E\left[ I_T1_{A^{(1,2,...,n)}_T }\right]$ will play the major role in the future price $F_{0,T}$. Here, we approximate this term and other terms are computed by the same manner. Using Taylor's theorem, we compute 
\begin{eqnarray} 
e^{V^k_0(T)X^k_T }&=&  1 + V^k_0(T)X^k_T + \frac{1}{2}e^{\xi^k_T} (V^k_0(T)X^k_T)^2, \label{eq:taylor}
\end{eqnarray}
where $\xi^k_T$ is between $0$ and  $ V^k_0(T)X^k_T$.  Using \eqref{eq:SVX} we write 
\begin{eqnarray}
E\left[ I_T1_{A^{(1,2,...,n)}_T }\right] &=& E\left[ 1_{A^{(1,2,...,n)}_T } \sum_{k=1}^{\overline{n}}w_0s^k_0e^{V^k_0(T)X^k_T}\right] \nonumber  \\
&=& E\left[ 1_{A^{(1,2,...,n)}_T } I_0 \right] +  E\left[ 1_{A^{(1,2,...,n)}_T } \sum_{k=1}^{\overline{n}}w_0s^k_0 V^k_0(T)X^k_T \right] \nonumber \\
&+& \frac{1}{2}  E\left[ 1_{A^{(1,2,...,n)}_T } \sum_{k=1}^{\overline{n}}w_0s^k_0 e^{\xi^k_T} (V^k_0(T)X^k_T)^2 \right]. \label{eq:F_A1} 
\end{eqnarray}
We consider the third term in \eqref{eq:F_A1}  and estimate 
\begin{eqnarray}
E[e^{\xi^k_T} (X^k_T)^2]  &\le&  E[(X^k_T)^2 1_{X^k_T \le 0}] + E[e^{V^k_0(T) X^k_T}(X^k_T)^21_{X^k_T >0}] \nonumber\\
&\le&E[(X^k_T)^2 1_{X^k_T \le 0}]  + E[e^{M^k_T - \frac{1}{2}\left\langle M^k \right\rangle_T} (X^k_T)^2]. \label{eq:exp_est} 
\end{eqnarray}
The first term in the RHS of   \eqref{eq:exp_est} is finite uniformly in $T$ for $T \in (0,1)$ by \eqref{eq:uniform}.  H\"older's inequality implies that
\begin{eqnarray*}
E[e^{M^k_T - \frac{1}{2}\left\langle M^k \right\rangle_T} (X^k_T)^2] &\le&  E^{1/p}[e^{pM^k_T - \frac{p^2}{2}\left\langle M^k  \right\rangle_T}]E^{1/p'}[ e^{(\frac{p'(p-1)}{2})\left\langle M^k  \right\rangle_T} (X^{k}_T)^{2p'}] \\
&\le& E^{1/p}[e^{pM^k_T - \frac{p^2}{2}\left\langle M^k  \right\rangle_T}]E^{1/p'}[ e^{(\frac{p'(p-1)}{2})\left\langle M^k  \right\rangle_T} (X^{k}_T)^{2p'}],
\end{eqnarray*}
where $1/p+1/p' = 1$ and $p > 1$.  Noting that $p'(p-1) = p$, we estimate
\begin{eqnarray*}
E[ e^{(\frac{p'(p-1)}{2})\left\langle M^k \right\rangle_T} (X^{k}_T)^{2p'}] \le E^{1/q}[ e^{\frac{qp}{2}\left\langle M^k \right\rangle_T} ] E^{1/q'} [(X^{k}_T)^{2p'q'}],
\end{eqnarray*}
where $1/q + 1/q' = 1, q > 1.$  We deduce from \eqref{eq:uniform} that 
\begin{equation}\label{eq:sup2}
\sup_{T \in (0,T^*)}E[e^{\xi^k_T} (X^{k}_T)^2] < \infty
\end{equation}
when $pq >1$ satisfies Assumption \ref{assumption:Gaussian_bound}. Therefore, the third term of \eqref{eq:F_A1} is of order $O(T)$. 

Using \eqref{eq:cond_V}, \eqref{eq:uniform} and then Lemma \ref{lemma:appro}, the second term of \eqref{eq:F_A1}  is approximated by
\begin{eqnarray*}
E\left[ 1_{A^{(1,2,...,n)}_T } \sum_{k=1}^{\overline{n}}w_0s^k_0 V^k_0(T)X^k_T \right] &=& E\left[ 1_{A^{(1,2,...,n)}_T } \sum_{k=1}^{\overline{n}}\nu_k\sqrt{T}X^k_T \right] +  \sum_{k=1}^{\overline{n}} O(T^{1/2+\zeta^k})   \\
&=& \int_{A^{(1,2,...,n)}_T} \left( \sum_{k=1}^{\overline{n}} \nu_kx_k \sqrt{T} \right)  q_t(\bx) d\bx +  \sum_{k=1}^{\overline{n}} O(T^{1/2+\zeta^k})  \\
&+&  \sqrt{T}\sum_{j=1}^n o(t^{\min\{2H^j,1\} +\varepsilon/4}),
\end{eqnarray*}
where we recall  $\nu_k = w_ks^k_0\sqrt{v^k_0(0)}$. By Theorem \ref{thm:density_expansion}, we write
\begin{eqnarray}
&&\int_{A^{(1,2,...,n)}_T} \left( \sum_{k=1}^{\overline{n}} \nu_kx_k \sqrt{T} \right)  q_t(\bx) d\bx	 \nonumber\\
&=& \int_{A^{(1,2,...,n)}_T} {\left( \sum_{k=1}^{ \overline{n}} \nu_kx_k\sqrt{T}   \right)  \phi_{\mathbf{\mu},\Gamma}\left( \bx \right) dx_n ... dx_1} \nonumber \\
&-&  \int_{A^{(1,2,...,n)}_T} \left( \sum_{k=1}^{ \overline{n}} \nu_kx_k\sqrt{T}   \right)   \left(\sum_{j=1}^nT^{H^j} a^{(1),j}_T(\bx) \right) dx_n ... dx_1 \nonumber \\
&-&  \int_{A^{(1,2,...,n)}_T} \left( \sum_{k=1}^{ \overline{n}} \nu_kx_k\sqrt{T}   \right)    \left(\sum_{j=1}^n  T^{2H^j}\left( \frac{1}{2}a^{(2),j}_T(\bx)  + c^j_T(\bx) \right) \right)  dx_n ... dx_1 \nonumber\\
&+&  \int_{A^{(1,2,...,n)}_T} \left( \sum_{k=1}^{ \overline{n}} \nu_kx_k\sqrt{T}   \right)    \left(\sum_{1 \le j, \ell \le n}^n  T^{H^k + H^j} d^{(1),kj}(\bx)  \right) dx_n ... dx_1 + O(T).  \label{eq:first_term}
\end{eqnarray}
For all  $2 \le j \le n$, it is clear that when $T \to 0$,
$$\frac{1}{V^j_0(T)}\log\left( \frac{s^{j-1}_0}{s^j_0}  \right) \sim \frac{1}{\sqrt{T}}  \to \infty, \qquad \lim_{T \to 0} \frac{V^{j-1}_0(T)}{V^j_0(T)} = \sqrt{\frac{v^{j-1}_0(0)}{v^j_0(0)}}.$$ 
Using Lemma \ref{lemma:tail1}, we estimate
\begin{eqnarray*}
&&\int_{A^{(1,2,...,n)}_T} {\left( \sum_{k=1}^{ \overline{n}} \nu_kx_k   \right)  \phi_{\mathbf{\mu},\Gamma}\left( \bx\right) dx_n ... dx_1} \\
&=& \int_{-\infty}^{\infty}\int_{-\infty}^{\frac{1}{V^2_0(T)}\log\left( \frac{s^1_0}{s^2_0}  \right)  + \frac{V^1_0(T)}{V^2_0(T)} x_1 } \cdots \int_{-\infty}^{\frac{1}{V^n_0(T)}\log\left( \frac{s^{n-1}_0}{s^n_0}  \right) +  \frac{V^{n-1}_0(T)}{V^n_0(T)}  x_{n-1} }\left( \sum_{k=1}^{ \overline{n}} \nu_kx_k   \right)  \phi_{\mathbf{\mu},\Gamma}\left( \bx \right) dx_n ... dx_1 \\
&=&  \int_{\mathbb{R}^n}  {(...)dx_n ... dx_1}  -  \int_{-\infty}^{\infty} \cdots \int_{\frac{1}{V^n_0(T)}\log\left( \frac{s^{n-1}_0}{s^n_0}  \right) +  \frac{V^{n-1}_0(T)}{V^n_0(T)}  x_{n-1} }^{\infty}{(...)dx_n ... dx_1} \\
&=& \int_{\mathbb{R}^n} \left( \sum_{k=1}^{ \overline{n}} \nu_kx_k   \right)  \phi_{\mathbf{\mu},\Gamma}\left( \bx \right) dx_n ... dx_1   + o(T).
\end{eqnarray*}
For $A^{\psi_n}_T$ when $\psi_n$ differs from $(1,2,...,n)$, Lemma \ref{lemma:tail1} implies that
\begin{eqnarray*}
\int_{A^{\psi_n}_T} {\left( \sum_{k=1}^{ \overline{n}} \nu_{\psi_n(k)}x_{\psi_n(k)}   \right)  q_t\left( \bx\right) d\bx} = o(T).
\end{eqnarray*}
For example, when $\psi_n = (2,1,3,...,n)$ we obtain
\begin{eqnarray*}
&&\int_{A^{\psi_n}_T} {\left( \sum_{k=1}^{ \overline{n}} \nu_{\psi_n(k)}x_{\psi_n(k)}   \right)  q_t\left( \bx\right) d\bx} \\
&=& \int_{-\infty}^{\infty}\int_{-\infty}^{\frac{1}{V^1_0(T)}\log\left( \frac{s^2_0}{s^1_0}  \right)  + \frac{V^2_0(T)}{V^1_0(T)} x_2 } \cdots \int_{-\infty}^{\frac{1}{V^n_0(T)}\log\left( \frac{s^{n-1}_0}{s^n_0}  \right) +  \frac{V^{n-1}_0(T)}{V^n_0(T)}  x_{n-1} }\left( ... \right) dx_n ... dx_3dx_1dx_2 = o(T),
\end{eqnarray*}
since 
$$\frac{1}{V^1_0(T)}\log\left( \frac{s^{2}_0}{s^1_0}  \right) \sim - \frac{1}{\sqrt{T}}  \to - \infty.$$ 
The conclusion follows by taking all permutations $\psi_n \in \Pi_n$ into account. 
\subsection{Proof of Proposition \ref{pro:future2}}\label{proof:future2}
Again, the future price is given by \eqref{eq:future_price}. The terms  $E[I_T1_{A^{(1,...,r-1,r,...n)}_T}],  E[I_T1_{A^{(1,...,r,r-1,...n)}_T}]$ are the most significant factors in the future price, where 
\begin{eqnarray*}
A^{(1,...,r-1,r,...n)}_T &=&  \{ \omega: S^{1}_T \ge ... \ge  S^{r-1}_T \ge S^r_T \ge ... \ge S^{n}_T \}, \\
A^{(1,...,r,r-1,...n)}_T &=&  \{ \omega: S^{1}_T \ge... \ge  S^{r}_T \ge S^{r-1}_T \ge ... \ge S^{n}_T \}. 
\end{eqnarray*}
It suffices to consider the event $	A^{(1,...,r-1,r,...n)}_T$. Using the argument with Taylor's theorem in Subsection \ref{proof:future1}, we arrive at the formula \eqref{eq:F_A1}
and the third term of \eqref{eq:F_A1} is also of order $O(T)$. Using \eqref{eq:cond_V}, \eqref{eq:uniform} and then Lemma \ref{lemma:appro}, the second term of \eqref{eq:F_A1}  is approximated by
\begin{eqnarray*}
E\left[ 1_{A^{(1,...,r-1,r,...,n)}_T } \sum_{k=1}^{\overline{n}}w_0s^k_0 V^k_0(T)X^k_T \right] &=& E\left[ 1_{A^{(1,...,r-1,r,...,n)}_T } \sum_{k=1}^{\overline{n}}\nu_k\sqrt{T}X^k_T \right] +  \sum_{k=1}^{\overline{n}} O(T^{1/2+\zeta^k})   \\
&=& \int_{A^{(1,...,r-1,r,...,n)}_T} \left( \sum_{k=1}^{\overline{n}} \nu_kx_k \sqrt{T} \right)  q_t(\bx) d\bx +  \sum_{k=1}^{\overline{n}} O(T^{1/2+\zeta^k})  \\
&+&  \sqrt{T}\sum_{j=1}^n o(t^{\min\{2H^j,1\} +\varepsilon/4}),
\end{eqnarray*}
Again from Theorem \ref{thm:density_expansion}, we obtain \eqref{eq:first_term} and the term 
\begin{eqnarray*}
&&\int_{A^{(1,...,r-1,r,...,n)}_T} {\left( \sum_{k=1}^{ \overline{n}} \nu_kx_k   \right)  \phi_{\mathbf{\mu},\Gamma}\left( \bx \right) dx_n ... dx_1} \\
&=&  \int_{-\infty}^{\infty} \cdots \int_{-\infty}^{\frac{1}{V^r_0(T)} \log\left( \frac{s^{r-1}_0}{s^r_0}\right) + \frac{V^{r-1}_0(T)}{V^r_0(T)} x_{r-1}} \cdots \int_{-\infty}^{\frac{1}{V^n_0(T)}\log\left( \frac{s^{n-1}_0}{s^n_0}  \right) +  \frac{V^{n-1}_0(T)}{V^n_0(T)}  x_{n-1} } {\left( \sum_{k=1}^{ \overline{n}} \nu_kx_k  \right)  \phi_{\mathbf{\mu},\Gamma}\left( \bx \right) dx_n ... dx_1}\\
&=& \int_{-\infty}^{\infty} \cdots \int_{-\infty}^{ \frac{V^{r-1}_0(T)}{V^r_0(T)} x_{r-1}} \cdots \int_{-\infty}^{\infty} {\left( \sum_{k=1}^{ \overline{n}} \nu_kx_k \right)  \phi_{\mathbf{\mu},\Gamma}\left( \bx \right) dx_n ... dx_1} + o(T), 
\end{eqnarray*}
by Lemma \ref{lemma:tail1} and noting that in this case, 
\begin{eqnarray*}
\frac{V^{r-1}_0(T)}{V^r_0(T)} &=& \sqrt{\frac{v^{r-1}_0(0)}{v^{r}_0(0)}} + O(T^{\zeta^{r-1}}) + O(T^{\zeta^r}),  \\
\frac{1}{V^i_0(T)}\log\left( \frac{s^{i-1}_0}{s^i_0} \right)&=& \frac{1}{\sqrt{T}}\to \infty \text{ for  } i \ne r. 
\end{eqnarray*}
The proof is complete by considering all permutations in $\Pi_n$. 
\subsection{Proof of Theorem \ref{thm:main}}\label{proof:thm:main}
Recall from \eqref{eq:C_k} that
\begin{equation*}
\frac{\partial C}{\partial k}(T,F_{0,T},k=0) = -F_{0,T}\mathbb{Q}\left( I_T > F_{0,T} \right) = - F_{0,T} \sum_{\psi_n\in\Pi_n}\mathbb{Q}\left( \{I_T > F_{0,T}\} \cap A^{\psi_n}_T\right), 
\end{equation*}
where the sets $A^{\psi_n}_T$ are defined in \eqref{eq:A}. 

Case $(i)$: It is enough to  consider the event $A^{(1,2,...,n)}_T$. Fixing
\begin{equation}\label{eq:condition_eta}
0 < \eta < \min \{\zeta^1,...,\zeta^n, 1/6\},
\end{equation} we estimate
\begin{eqnarray}
\mathbb{Q}\left( I_T > F_{0,T},A^{(1,...,n)}_T \right) &=& \mathbb{Q}\left( I_T > F_{0,T},A^{(1,...,n)}_T, \max_{k\in \{1,...,n\}}|X^k_T| \le \frac{1}{T^{\eta}}\right)  \nonumber \\
&+& \mathbb{Q}\left( I_T > F_{0,T},A^{(1,...,n)}_T, \max_{k\in \{1,...,n\}}|X^k_T|> \frac{1}{T^{\eta}} \right). \label{eq:de}
\end{eqnarray}
Choosing $p$ such that $p\eta > 1$, the second term of \eqref{eq:de} is bounded by
\begin{eqnarray}
\sum_{k=1}^n\mathbb{Q}\left( X^k_T \ge \frac{1}{T^{\eta}} \right) \le T^{p\eta} \sum_{k=1}^nE[|X^k_T|^p,	\label{eq:big}
\end{eqnarray} 
which is of order $O(T)$ from  \eqref{eq:uniform}. We consider the first term of \eqref{eq:de}. For simple notation, we denote $\mathfrak{m}_1 := \sum_{k=1}^{\overline{n}} m^{k}_1. $  Using the expansion
\begin{eqnarray}
e^{V^k_0(T)X^k_T }=  1 + V^k_0(T)X^k_T +  \frac{1}{2}(V^k_0(T)X^k_T)^2 + \frac{1}{6}e^{\xi^k_T}(V^k_0(T)X^k_T)^3, \label{eq:exp}
\end{eqnarray} where $|\xi^k_T| \le |V^k_0(T)X^k_T|$ and the condition \eqref{eq:cond_V}, we find that 
\begin{eqnarray*}
I_T - I_0 &=&  \sum_{k=1}^{ \overline{n}} w_ks^k_0 e^{V^k_0(T)X^k_T} - I_0 \\
&=&  \sum_{k=1}^{ \overline{n}} w_ks^k_0 \left( V^k_0(T)X^k_T +  \frac{1}{2}(V^k_0(T)X^k_T)^2 + \frac{1}{6}e^{\xi^k_T}(V^k_0(T)X^k_T)^3 \right) \\
&=& \sum_{k=1}^{ \overline{n}}  \left( \nu_kT^{1/2}X^k_T + \frac{1}{2} \nu_k\sqrt{v^k_0(0)}T (X^k_T)^2 + w_ks^k_0 \frac{1}{6}e^{\xi^k_T}(V^k_0(T)X^k_T)^3 \right) + \sum_{k=1}^{ \overline{n}} O(T^{1/2 + \zeta^k-\eta}).
\end{eqnarray*}
Therefore, 
\begin{eqnarray}
\mathfrak{B}(T,\mathbf{X})&:=& \frac{1}{\sqrt{T}} \left(I_T - I_0 -\sum_{k=1}^{ \overline{n}}  \nu_kT^{1/2}X^k_T - \sum_{k=1}^{ \overline{n}}\frac{1}{2} \nu_k\sqrt{v^k_0(0)}T (X^k_T)^2\right) \nonumber  \\
&=& \sum_{k=1}^{ \overline{n}} O(T^{\zeta^k-\eta}) +O(T^{1-3\eta}), \label{eq:bound_B}
\end{eqnarray} 
on the event $A^{(1,2...,n)}_T \bigcap \{\max_{k\in \{1,...,n\}}|X^k_T| \le \frac{1}{T^{\eta}}\}$. Then  Proposition \ref{pro:future1} yields 
\begin{eqnarray*}
\frac{I_T-F_{0,T}}{\sqrt{T}}    &=& \frac{(I_T - I_0)}{\sqrt{T}} + \frac{(I_0 - F_{0,T} )}{\sqrt{T}} \\
&=&  \underbrace{\left( \sum_{k=1}^{ \overline{n}} \nu_kX^k_T  + T^{1/2}\sum_{k=1}^{ \overline{n}} \nu_k\sqrt{v^k_0(0)}(X^k_T)^2 -  \mathfrak{m}_1  \right)}_{:=B(T,\bX)}  - \gamma^1(T, \bX),
\end{eqnarray*}
where 
\begin{eqnarray*}
\gamma^1(T,\mathbf{X})&:=& - \mathfrak{B}(T,\mathbf{X}) +  \sum_{1 \le k  \le \overline{n}, 1\le j \le n} m^{k,j}_2T^{H^j}  +  \sum_{1 \le k  \le \overline{n}, 1\le j \le n} m^{k,j}_3T^{2H^j}\\
&& +  \sum_{1 \le k  \le \overline{n}, 1\le j, \ell \le n} m^{k,j,\ell}_4 T^{H^k + H^j} + O(T^{1/2}) + \sum_{k=1}^{\overline{n}} O(T^{ \zeta^k}) + \sum_{j=1}^n o(T^{\min\{2H^j,1\} +\varepsilon/4}).
\end{eqnarray*}

We claim that 
\begin{eqnarray}
	&&\mathbb{Q}\left( I_T > F_{0,T},A^{(1,...,n)}_T, \max_{k\in \{1,...,n\}}|X^k_T| \le \frac{1}{T^{\eta}}\right)  - \mathbb{Q}\left( B(T,\bX) > 0,A^{(1,...,n)}_T, \max_{k\in \{1,...,n\}}|X^k_T| \le \frac{1}{T^{\eta}}  \right) \nonumber \\
	&=& O(\gamma^1(T)). \label{eq:claim}
\end{eqnarray}
Indeed, from the identities
\begin{eqnarray*}
&&\{ 0 < B(T,\bX) < \gamma^1(T,\mathbf{X})\} \cup \{ B(T,\bX) > \gamma^1(T,\mathbf{X}), \gamma^1(T,\mathbf{X}) > 0\}  = \{ B(T,\bX) > 0, \gamma^1(T,\mathbf{X}) > 0 \},\\
&& \{ 0 > B(T,\bX) > \gamma^1(T,\mathbf{X})\} \cup \{ B(T,\bX) > 0, \gamma^1(T,\mathbf{X}) < 0 \} = \{B(T,\bX) > \gamma^1(T,\mathbf{X}), \gamma^1(T,\mathbf{X}) < 0\}
\end{eqnarray*} 
we deduce that 
\begin{eqnarray}
&&\left| \mathbb{Q}\left( I_T > F_{0,T},A^{(1,...,n)}_T, \max_{k\in \{1,...,n\}}|X^k_T| \le \frac{1}{T^{\eta}}\right)  - \mathbb{Q}\left( B(T,\bX) > 0,A^{(1,...,n)}_T, \max_{k\in \{1,...,n\}}|X^k_T| \le \frac{1}{T^{\eta}}  \right) \right|  \nonumber \\
&\le& \mathbb{Q}\left( \gamma^1(T,\bX) < B(T,\bX) <0, ,A^{(1,...,n)}_T ,\max_{k\in \{1,...,n\}}|X^k_T| \le \frac{1}{T^{\eta}} \right) \nonumber \\
&+& \mathbb{Q}\left( 0 < B(T,\bX) < \gamma^1(T,\bX), ,A^{(1,...,n)}_T ,\max_{k\in \{1,...,n\}}|X^k_T| \le \frac{1}{T^{\eta}} \right). \label{eq:claim2}
\end{eqnarray}
Using \eqref{eq:bound_B}, the third quantity of \eqref{eq:claim2} is bounded by
$$\mathbb{Q}\left( \gamma^1(T) < B(T,\bX) <0 ,A^{(1,...,n)}_T ,\max_{k\in \{1,...,n\}}|X^k_T| \le \frac{1}{T^{\eta}} \right).$$
where
\begin{eqnarray}
\gamma^1(T)&:=& \sum_{1 \le k  \le \overline{n}, 1\le j \le n} m^{kj}_2T^{H^j}  +  \sum_{1 \le k  \le \overline{n}, 1\le j \le n} m^{kj}_3T^{2H^j} +  \sum_{1 \le k  \le \overline{n}, 1\le j, \ell \le n} m^{kj\ell}_4 T^{H^k + H^j}  \nonumber \\
&+& O(T^{1/2}) + \sum_{j=1}^n o(T^{\min\{2H^j,1\} +\varepsilon/4}) + \sum_{k=1}^{ \overline{n}} O(T^{\zeta^k-\eta}) +O(T^{1-3\eta}).\label{eq:gamma1}
\end{eqnarray}
Let $\Xi(T) \in \{0, \gamma^1(T)\}$. If $T$ is small enough, the equation
$$T^{1/2}\nu_1\sqrt{v^1_0(0)}x^2_1 + \nu_1x_1 + \left( \sum_{k=2}^{ \overline{n}} \nu_kx_k + T^{1/2} \sum_{k=2}^{ \overline{n}} \nu_k\sqrt{v^k_0(0)}x^2_k -  \mathfrak{m}_1 - \Xi(T) \right) = 0$$
has always two solutions
$$\chi^{\Xi(T)}_{\pm} = \frac{-1\pm\sqrt{1-4\nu^{-1}_1\sqrt{v^1_0(0)}\sqrt{T}\left( \sum_{k=2}^{ \overline{n}} \nu_kx_k + T^{1/2} \sum_{k=2}^{ \overline{n}} \nu_k\sqrt{v^k_0(0)}x^2_k -  \mathfrak{m}_1 - \Xi(T) \right) }}{2\sqrt{T}\sqrt{v^1_0(0)}} $$
because $\max_{k\in \{1,...,n\}}|X^k_T| \le \frac{1}{T^{\eta}}$.  Applying the expansion $\sqrt{1+x} = 1 + x/2 + O(x^2)$, if $T$ is small enough then $\chi^{\Xi(T)}_{-} \sim -1/\sqrt{T}$ and
\begin{eqnarray}\label{eq:chi}
 \chi^{\Xi(T)}_{+} = \nu^{-1}_1\left( \sum_{k=2}^{ \overline{n}} \nu_kx_k + T^{1/2} \sum_{k=2}^{ \overline{n}} \nu_k\sqrt{v^k_0(0)}x^2_k -  \mathfrak{m}_1 - \Xi(T)\right) + O(\sqrt{T}), 
\end{eqnarray}
and we arrive at
$$|\chi^{\gamma^1(T)}_+ -  \chi^0_+| = |\chi^{\gamma^1(T)}_- -  \chi^0_-|  = O(\gamma^1(T)). $$
Therefore, Lemma \ref{lemma:tail1} and Lemma \ref{lemma:appro} yield
\begin{eqnarray*}
&&\mathbb{Q}\left( \gamma^1(T) < B(T,\bX) <0, , A^{(1,...,n)}_T,\max_{k\in \{1,...,n\}}|X^k_T| \le \frac{1}{T^{\eta}} \right) \\
&=&  \int_{\mathbb{R}^{n-1}}  \int_{\mathbb{R}} 1_{ [\chi^{\gamma^1(T)}_+, \chi^0_+] \bigcup [\chi^{\gamma^1(T)}_-, \chi^0_-]}(x_1) 1_{\max_{k\in \{1,...,n\}}|x^k| \le \frac{1}{T^{\eta}}} q_t(\bx) dx_1 ... dx_n  = O(\gamma^1(T)).
\end{eqnarray*}
The same argument holds true for the fourth quantity of \eqref{eq:claim2} and hence, the claim \eqref{eq:claim} holds.  

Using similar arguments and Lemma \ref{lemma:appro}, we compute  
\begin{eqnarray*}
&&\mathbb{Q}\left( B(T,\bX) > 0,A^{(1,...,n)}_T, \max_{k\in \{1,...,n\}}|X^k_T| \le \frac{1}{T^{\eta}}  \right)\\
&&\qquad =  \int_{\mathbb{R}^{n-1}}  \int_{\chi^0_+}^{\infty} 1_{\max_{k\in \{1,...,n\}}|x^k| \le \frac{1}{T^{\eta}}}q_T(\bx)d\bx \\
&& \qquad + \int_{\mathbb{R}^{n-1}} \int_{-\infty}^{\chi^0_-} 1_{\max_{k\in \{1,...,n\}}|x^k| \le \frac{1}{T^{\eta}}}q_T(\bx)d\bx + \sum_{j=1}^n o(T^{\min\{2H^j,1\} +\varepsilon/4}).
\end{eqnarray*}
If $T$ is small enough such that $2 \sqrt{v^1_0(0)}\sqrt{T} < T^{\eta}$, we have that $\chi^0_- \le -1/(2\sqrt{v^1_0(0)}\sqrt{T})$ and thus
\begin{eqnarray*}
\int_{\mathbb{R}^{n-1}}  \int_{-\infty}^{\chi^0_-} 1_{\max_{k\in \{1,...,n\}}|x^k| \le \frac{1}{T^{\eta}}}q_T(\bx)d\bx  \le \int_{\mathbb{R}^{n-1}}  \int_{-\infty}^{-1/(2 \sqrt{v^1_0(0)}\sqrt{T})} 1_{\max_{k\in \{1,...,n\}}|x^k|
\le \frac{1}{T^{\eta}}}q_T(\bx)d\bx = o(T).
\end{eqnarray*}
by Lemma \ref{lemma:tail1}. Using \eqref{eq:chi}, we arrive at
\begin{eqnarray}
\int_{\chi^0_+}^{\infty} q_T(\bx)dx_1 &=&  \int_{\infty }^{\infty} 1_{  \sum_{k=1}^{ \overline{n}} \nu_kx_k > \mathfrak{m}_1 } q_T(\bx)dx_1 \nonumber \\
&& \qquad + \int_{-\nu_1^{-1}\left( \sum_{k=2}^{ \overline{n}} \nu_kx_k + T^{1/2} \sum_{k=2}^{ \overline{n}} \nu_k\sqrt{v^k_0(0)}x^2_k- \mathfrak{m}_1 \right) }^{-\nu_1^{-1}\left(\sum_{k=2}^{ \overline{n}} \nu_kx_k - \mathfrak{m}_1\right) }q_T(\bx)dx_1. \label{eq:chi+1}
\end{eqnarray}
By the mean value theorem, the second integral of \eqref{eq:chi+1} is bounded by 
$$T^{1/2} \nu_1^{-1} \left(  \sum_{k=2}^{ \overline{n}} \nu_k\sqrt{v^k_0(0)}x^2_k\right) q_T(\xi^1,x_2,...x_n)$$ for some 
$$ \xi^1 \in \left[ -\nu_1^{-1}\left(\sum_{k=2}^{ \overline{n}} \nu_kx_k - \mathfrak{m}_1\right),  -\nu_1^{-1}\left( \sum_{k=2}^{ \overline{n}} \nu_kx_k + T^{1/2} \sum_{k=2}^{ \overline{n}} \nu_k\sqrt{v^k_0(0)}x^2_k- \mathfrak{m}_1\right) \right].$$
Therefore,
\begin{eqnarray*}
\int_{\mathbb{R}^{n-1}} \int_{\chi_+}^{\infty} 1_{\max_{k\in \{1,...,n\}}|x^k| \le \frac{1}{T^{\eta}}}  q_T(\bx)d\bx = \int_{\mathbb{R}^n}1_{\sum_{k=1}^n \nu_k x_k \ge \mathfrak{m}_1} q_T(x_1,..,x_n)dx_1 +  O(T^{1/2}).  
\end{eqnarray*}
Finally, we obtain that 
\begin{eqnarray*}
\frac{\partial C}{\partial k}(0,T,k=0)  &=& -F_{0,T}\left( \int_{D^1} q_T(\bx)  d\bx  + O(\gamma^1(T))\right) ,
\end{eqnarray*}
where 
\begin{equation}\label{eq:D1}
D^1 := \left\lbrace \bx \in \mathbb{R}^n:  \sum_{k=1}^{ \overline{n}} \nu_kx_k  > \mathfrak{m}_1  \right\rbrace 
\end{equation}
and the conclusion for this case follows.

Case $(ii)$: It suffices to consider the sum of the two events
\begin{eqnarray*}
\mathfrak{q}^1_T &=& \mathbb{Q}\left( I_T > F_{0,T},A^{\psi^1}_T\right),\\
\mathfrak{q}^2_T &=&  \mathbb{Q}\left( I_T > F_{0,T},A^{\psi^2}_T\right),
\end{eqnarray*}  
with the corresponding permutations $\psi^1_n = (1,...,r-1,r,..,n)$ and $\psi^2_n = (1,...,r,r-1,...,n)$. Following the same arguments as in Case 1, we have 
\begin{eqnarray*}
\mathfrak{q}^1_T&=&  \int_{-\infty}^{\infty} \cdots \int_{-\infty}^{\sqrt{\frac{v^{r-1}_0(0)}{v^r_0(0)}} x_{r-1}} \cdots \int_{-\infty}^{\infty}  1_{\sum_{k=1}^{ \overline{n}} \nu_kX^k_T  > \sum_{k = 1}^{\overline{n}} m^k_5} q_T(\bx)  dx_n... dx_1 \\
&+& O(T^{\zeta^{r-1}}) + O(T^{\zeta^{r}}) + O(\gamma^2(T)), \label{eq:112}
\end{eqnarray*}
where 
\begin{eqnarray}
\gamma^{2}(T)&:=&  \sum_{1\le k \le \overline{n}, 1 \le j \le n} m^{k,j}_6T^{H^j} +\sum_{1\le k \le \overline{n}, 1 \le j \le n} m^{k,j}_7T^{2H^j} + \sum_{1\le k \le \overline{n}, 1 \le j,\ell \le n} m^{k,j,\ell}_8T^{H^j+H^{\ell}} \nonumber \\
 &+&O(\sqrt{T}) + \sum_{j=1}^n o(T^{\min\{2H^j,1\} +\varepsilon/4}) + \sum_{k=1}^{ \overline{n}} O(T^{\zeta^k-\eta}) +  O(T^{1-3\eta}). \label{eq:gamma2}
\end{eqnarray}
A similar formula holds for $\mathfrak{q}^2_T$. Therefore,
\begin{eqnarray*}
\mathfrak{q}^1_T + \mathfrak{q}^2_T  &=& \int_{D^{2,1}\cup D^{2,2}} q_T(\bx)  d\bx  + O(T^{\zeta^{r-1}}) + O(T^{\zeta^{r}})+ O(\gamma^2(T)), \nonumber
\end{eqnarray*}
where
\begin{eqnarray}
D^{2,1} &=& \left\lbrace \bx \in \mathbb{R}^n:  \sum_{k=1}^{ \overline{n}} w_ks^{\psi^1_n(k)}_0\sqrt{v^k_0(0)}x_{\psi^1_n(k)}  > \sum_{j=1}^{\overline{n}} m^k_5  \text{ and } \sqrt{v^r_0(0)}x_r \le \sqrt{v^{r-1}_0(0)} x_{r-1} \right\rbrace, \label{eq:D21}\\
D^{2,2} &=& \left\lbrace \bx \in \mathbb{R}^n:  \sum_{k=1}^{ \overline{n}} w_ks^{\psi^2_n(k)}_0\sqrt{v^k_0(0)}x_{\psi^2_n(k)}  > \sum_{j=1}^{\overline{n}} m^k_5  \text{ and }  \sqrt{v^{r-1}_0(0)} x_{r-1} \le \sqrt{v^r_0(0)}x_r  \right\rbrace.  \label{eq:D22}
\end{eqnarray}
The proof is complete. 
\section{Appendix}\label{sec:app}
We provide some useful formulas. 
\begin{lemma}\label{lemma:tail1}
Let $f$ be a real-valued  function in the Schwartz space.  Then $\int_z^{\infty}f(x)dx = O(z^{-r})$, where we could choose any $r \ge 1$. 
\end{lemma}
\begin{proof}
For any $r \ge 2$, we know that $\sup_{x \in \mathbb{R}} |x^rf(x)| < C$ for some $C > 0$. We estimate easily that 
\begin{eqnarray*}
\int_{z}^{\infty} f(x)dx \le C \int_z^{\infty} x^{-r} dx \sim z^{-r+1},
\end{eqnarray*}
and the conclusion follows. 
\end{proof}
\begin{lemma}\label{lemma:intbypart}
Let $f(x): \mathbb{R} \to \mathbb{R}$ be a function vanishing at $-\infty$ and $\infty$. Then
\begin{eqnarray*}
- \int e^{i\sum_{j=1}^nu_jx_j}\sum_{j=1}^ n iu_jf_j(x)dx &=& \int e^{i\sum_{j=1}^nu_jx_j} \sum_{j=1}^n \frac{\partial f_j}{\partial x_j}dx,\\
\int e^{i \sum_{j=1}^n u_jx_j} \frac{\partial^n}{\partial x_1... \partial x_n}f(x)dx &=& \int e^{i\sum_{j=1}^n u_jx_j} i^nu_1...u_nf(x)dx.
\end{eqnarray*}
\end{lemma}
\begin{proof}
Recall the integration by parts
$$
\int_{\mathbb R^n} D^{\alpha}f(x)g( x)d x
=(-1)^{|\alpha|}\int_{\mathbb R^n}f( x)D^{\alpha}g( x)d x.
$$
The proofs of the two identities follows
immediately from this. For example, we have 
\begin{align*}
\int e^{i\sum_{j=1}^nu_jx_j} \sum_{j=1}^n \frac{\partial f_j}{\partial x_j}dx&=\sum_{j=1}^n
\int e^{i\sum_{j=1}^nu_jx_j}  \frac{\partial f_j}{\partial x_j}dx\\
&=-\sum_{j=1}^n
\int e^{i\sum_{j=1}^nu_jx_j}  iu_j f_j (x) dx=-\int e^{i\sum_{j=1}^nu_jx_j}\sum_{j=1}^ n iu_jf_j(x)dx .
\end{align*}
\end{proof}


\begin{thebibliography}{}
%
%
\bibitem[Abi~Jaber, 2019]{abi2019lifting}
Abi~Jaber, E. (2019).
\newblock Lifting the {H}eston model.
\newblock {\em Quantitative Finance}, 19(12):1995--2013.

\bibitem[Abi~Jaber, 2022]{abi2022characteristic}
Abi~Jaber, E. (2022).
\newblock The characteristic function of gaussian stochastic volatility models:
an analytic expression.
\newblock {\em Finance and Stochastics}, 26(4):733--769.

\bibitem[Alos et~al., 2007]{alos2007short}
Alos, E., Le{\'o}n, J.~A., and Vives, J. (2007).
\newblock On the short-time behavior of the implied volatility for
jump-diffusion models with stochastic volatility.
\newblock {\em Finance and Stochastics}, 11(4):571--589.

\bibitem[Andersen and Bollerslev, 1997]{andersen1997intraday}
Andersen, T.~G. and Bollerslev, T. (1997).
\newblock Intraday periodicity and volatility persistence in financial markets.
\newblock {\em Journal of Empirical Finance}, 4(2-3):115--158.

\bibitem[Andersen et~al., 2003]{andersen2003modeling}
Andersen, T.~G., Bollerslev, T., Diebold, F.~X., and Labys, P. (2003).
\newblock Modeling and forecasting realized volatility.
\newblock {\em Econometrica}, 71(2):579--625.

\bibitem[Avellaneda et~al., 2003]{avellaneda2003application}
Avellaneda, M., Boyer-Olson, D., Friz, P., et~al. (2003).
\newblock Application of large deviation methods to the pricing of index
options in finance.
\newblock {\em Comptes rendus. Math{\'e}matique}, 336(3):263--266.

\bibitem[Banner and Ghomrasni, 2008]{banner2008local}
Banner, A.~D. and Ghomrasni, R. (2008).
\newblock Local times of ranked continuous semimartingales.
\newblock {\em Stochastic Processes and their Applications}, 118(7):1244--1253.

\bibitem[Barletta et~al., 2019]{barletta2019short}
Barletta, A., Nicolato, E., and Pagliarani, S. (2019).
\newblock The short-time behavior of {VIX}-implied volatilities in a
multifactor stochastic volatility framework.
\newblock {\em Mathematical Finance}, 29(3):928--966.

\bibitem[Bayer et~al., 2016]{bayer2016pricing}
Bayer, C., Friz, P., and Gatheral, J. (2016).
\newblock Pricing under rough volatility.
\newblock {\em Quantitative Finance}, 16(6):887--904.

\bibitem[Bayer et~al., 2020]{bayer2020regularity}
Bayer, C., Friz, P.~K., Gassiat, P., Martin, J., and Stemper, B. (2020).
\newblock A regularity structure for rough volatility.
\newblock {\em Mathematical Finance}, 30(3):782--832.

\bibitem[Bayer et~al., 2019]{bayer2019short}
Bayer, C., Friz, P.~K., Gulisashvili, A., Horvath, B., and Stemper, B. (2019).
\newblock Short-time near-the-money skew in rough fractional volatility models.
\newblock {\em Quantitative Finance}, 19(5):779--798.

\bibitem[Bayer and Laurence, 2014]{bayer2014asymptotics}
Bayer, C. and Laurence, P. (2014).
\newblock Asymptotics beats {M}onte {C}arlo: The case of correlated local vol
baskets.
\newblock {\em Communications on Pure and Applied Mathematics},
67(10):1618--1657.

\bibitem[Bennedsen et~al., 2022]{bennedsen2022decoupling}
Bennedsen, M., Lunde, A., Pakkanen, and S, M. (2022).
\newblock Decoupling the short-and long-term behavior of stochastic volatility.
\newblock {\em Journal of Financial Econometrics}, 20(5):961--1006.

\bibitem[Bennedsen et~al., 2017]{bennedsen2017hybrid}
Bennedsen, M., Lunde, A., and Pakkanen, M.~S. (2017).
\newblock Hybrid scheme for {B}rownian semistationary processes.
\newblock {\em Finance and Stochastics}, 21:931--965.

\bibitem[Berestycki et~al., 2004]{berestycki2004computing}
Berestycki, H., Busca, J., and Florent, I. (2004).
\newblock Computing the implied volatility in stochastic volatility models.
\newblock {\em Communications on Pure and Applied Mathematics: A Journal Issued
by the Courant Institute of Mathematical Sciences}, 57(10):1352--1373.

\bibitem[Black and Scholes, 1973]{black1973pricing}
Black, F. and Scholes, M. (1973).
\newblock The pricing of options and corporate liabilities.
\newblock {\em Journal of Political Economy}, 81(3):637--654.

\bibitem[Comte et~al., 2012]{comte2012affine}
Comte, F., Coutin, L., and Renault, E. (2012).
\newblock Affine fractional stochastic volatility models.
\newblock {\em Annals of Finance}, 8:337--378.

\bibitem[Comte and Renault, 1998]{comte1998long}
Comte, F. and Renault, E. (1998).
\newblock Long memory in continuous-time stochastic volatility models.
\newblock {\em Mathematical Finance}, 8(4):291--323.

\bibitem[Cont, 2007]{cont2007volatility}
Cont, R. (2007).
\newblock Volatility clustering in financial markets: empirical facts and
agent-based models.
\newblock In {\em Long Memory in Economics}, pages 289--309. Springer.

\bibitem[Cont and Das, 2024]{cont2022rough}
Cont, R. and Das, P. (2024).
\newblock Rough volatility: fact or artefact?
\newblock {\em Sankhya B}, pages 1--33.

\bibitem[Ding et~al., 1993]{ding1993long}
Ding, Z., Granger, C.~W., and Engle, R.~F. (1993).
\newblock A long memory property of stock market returns and a new model.
\newblock {\em Journal of Empirical Finance}, 1(1):83--106.

\bibitem[Durrleman, 2010]{durrleman2010implied}
Durrleman, V. (2010).
\newblock From implied to spot volatilities.
\newblock {\em Finance and Stochastics}, 14:157--177.

\bibitem[El~Euch et~al., 2018]{el2018microstructural}
El~Euch, O., Fukasawa, M., and Rosenbaum, M. (2018).
\newblock The microstructural foundations of leverage effect and rough
volatility.
\newblock {\em Finance and Stochastics}, 22:241--280.

\bibitem[El~Euch and Rosenbaum, 2019]{el2019characteristic}
El~Euch, O. and Rosenbaum, M. (2019).
\newblock The characteristic function of rough {H}eston models.
\newblock {\em Mathematical Finance}, 29(1):3--38.

\bibitem[El~Euch et~al., 2019]{euch2019short}
El~Euch, O.~E., Fukasawa, M., Gatheral, J., and Rosenbaum, M. (2019).
\newblock Short-term at-the-money asymptotics under stochastic volatility
models.
\newblock {\em SIAM Journal on Financial Mathematics}, 10(2):491--511.

\bibitem[Fernholz, 2002]{fernholz2002stochastic}
Fernholz, E.~R. (2002).
\newblock {\em Stochastic portfolio theory}.
\newblock Springer.

\bibitem[Fernholz, 1999]{fernholz1999diversity}
Fernholz, R. (1999).
\newblock On the diversity of equity markets.
\newblock {\em Journal of Mathematical Economics}, 31(3):393--417.

\bibitem[Fernholz, 2001]{fernholz2001equity}
Fernholz, R. (2001).
\newblock Equity portfolios generated by functions of ranked market weights.
\newblock {\em Finance and Stochastics}, 5:469--486.

\bibitem[Forde et~al., 2021]{forde2021small}
Forde, M., Gerhold, S., and Smith, B. (2021).
\newblock Small-time, large-time, and asymptotics for the rough {H}eston model.
\newblock {\em Mathematical Finance}, 31(1):203--241.

\bibitem[Forde and Jacquier, 2009]{forde2009small}
Forde, M. and Jacquier, A. (2009).
\newblock Small-time asymptotics for implied volatility under the {H}eston
model.
\newblock {\em International Journal of Theoretical and Applied Finance},
12(06):861--876.

\bibitem[Forde et~al., 2012]{forde2012small}
Forde, M., Jacquier, A., and Lee, R. (2012).
\newblock The small-time smile and term structure of implied volatility under
the {H}eston model.
\newblock {\em SIAM Journal on Financial Mathematics}, 3(1):690--708.

\bibitem[Forde and Zhang, 2017]{forde2017asymptotics}
Forde, M. and Zhang, H. (2017).
\newblock Asymptotics for rough stochastic volatility models.
\newblock {\em SIAM Journal on Financial Mathematics}, 8(1):114--145.

\bibitem[Friz et~al., 2021]{friz2021precise}
Friz, P.~K., Gassiat, P., and Pigato, P. (2021).
\newblock Precise asymptotics: robust stochastic volatility models.
\newblock {\em The Annals of Applied Probability}, 31(2):896--940.

\bibitem[Friz et~al., 2022]{friz2022short}
Friz, P.~K., Gassiat, P., and Pigato, P. (2022).
\newblock Short-dated smile under rough volatility: asymptotics and numerics.
\newblock {\em Quantitative Finance}, 22(3):463--480.

\bibitem[Friz and Wagenhofer, 2023]{friz2023reconstructing}
Friz, P.~K. and Wagenhofer, T. (2023).
\newblock Reconstructing volatility: Pricing of index options under rough
volatility.
\newblock {\em Mathematical Finance}.

\bibitem[Fukasawa, 2011]{fukasawa2011asymptotic}
Fukasawa, M. (2011).
\newblock Asymptotic analysis for stochastic volatility: martingale expansion.
\newblock {\em Finance and Stochastics}, 15:635--654.

\bibitem[Fukasawa, 2017]{fukasawa2017short}
Fukasawa, M. (2017).
\newblock Short-time at-the-money skew and rough fractional volatility.
\newblock {\em Quantitative Finance}, 17(2):189--198.

\bibitem[Fukasawa, 2021]{fukasawa2021volatility}
Fukasawa, M. (2021).
\newblock Volatility has to be rough.
\newblock {\em Quantitative Finance}, 21(1):1--8.

\bibitem[Fukasawa et~al., 2022]{fukasawa2022consistent}
Fukasawa, M., Takabatake, T., and Westphal, R. (2022).
\newblock Consistent estimation for fractional stochastic volatility model
under high-frequency asymptotics.
\newblock {\em Mathematical Finance}, 32(4):1086--1132.

\bibitem[Funahashi and Kijima, 2017]{funahashi2017does}
Funahashi, H. and Kijima, M. (2017).
\newblock Does the hurst index matter for option prices under fractional
volatility?
\newblock {\em Annals of Finance}, 13:55--74.

\bibitem[Gao and Lee, 2014]{gao2014asymptotics}
Gao, K. and Lee, R. (2014).
\newblock Asymptotics of implied volatility to arbitrary order.
\newblock {\em Finance and Stochastics}, 18:349--392.

\bibitem[Gatheral et~al., 2018]{gatheral2018volatility}
Gatheral, J., Jaisson, T., and Rosenbaum, M. (2018).
\newblock Volatility is rough.
\newblock {\em Quantitative Finance}, 18:933--949.

\bibitem[Gassiat, 2019]{gassiat} Gassiat, P. (2019). 
\newblock On the martingale property in the rough Bergomi model.
\newblock {\em Electronic Communications in Probability}, 24:1 -- 9. 

\bibitem[Gulisashvili, 2020]{gulisashvili2020gaussian}
Gulisashvili, A. (2020).
\newblock Gaussian stochastic volatility models: Scaling regimes, large
deviations, and moment explosions.
\newblock {\em Stochastic Processes and their Applications}, 130(6):3648--3686.

\bibitem[Gulisashvili and Tankov, 2015]{gulisashvili2015implied}
Gulisashvili, A. and Tankov, P. (2015).
\newblock Implied volatility of basket options at extreme strikes.
\newblock In {\em Large Deviations and Asymptotic Methods in Finance}, pages
175--212. Springer.

\bibitem[Gulisashvili et~al., 2019]{gulisashvili2019extreme}
Gulisashvili, A., Viens, F., and Zhang, X. (2019).
\newblock Extreme-strike asymptotics for general gaussian stochastic volatility
models.
\newblock {\em Annals of Finance}, 15(1):59--101.

\bibitem[Guyon and El~Amrani, 2022]{guyon2022does}
Guyon, J. and El~Amrani, M. (2022).
\newblock Does the term-structure of equity at-the-money skew really follow a
power law?
\newblock {\em Available at SSRN 4174538}.

\bibitem[Ichiba and Karatzas, 2010]{ichiba2010collisions}
Ichiba, T. and Karatzas, I. (2010).
\newblock On collisions of {B}rownian particles.
\newblock {\em Annals of Applied Probability}, 20(3):951--977.

\bibitem[Ichiba et~al., 2013]{ichiba2013convergence}
Ichiba, T., Pal, S., and Shkolnikov, M. (2013).
\newblock Convergence rates for rank-based models with applications to
portfolio theory.
\newblock {\em Probability Theory and Related Fields}, 156(1-2):415--448.

\bibitem[Ichiba et~al., 2011]{ichiba2011hybrid}
Ichiba, T., Papathanakos, V., Banner, A., Karatzas, I., and Fernholz, R.
(2011).
\newblock Hybrid atlas models.
\newblock {\em Annals of Applied Probability}, 21(2):609--644.

\bibitem[Jaisson and Rosenbaum, 2016]{jaisson2016rough}
Jaisson, T. and Rosenbaum, M. (2016).
\newblock Rough fractional diffusions as scaling limits of nearly unstable
heavy tailed {H}awkes processes.
\newblock {\em The Annals of Applied Probability}, 26(5):2860--2882.

\bibitem[Jiang and Wang, 2007]{jiang2007collision}
Jiang, Y. and Wang, Y. (2007).
\newblock On the collision local time of fractional {B}rownian motions.
\newblock {\em Chinese Annals of Mathematics, Series B}, 28(3):311--320.

\bibitem[Jourdain and Sbai, 2012]{jourdain2012coupling}
Jourdain, B. and Sbai, M. (2012).
\newblock Coupling index and stocks.
\newblock {\em Quantitative Finance}, 12(5):805--818.

\bibitem[Karatzas and Fernholz, 2009]{karatzas2009stochastic}
Karatzas, I. and Fernholz, R. (2009).
\newblock Stochastic portfolio theory: an overview.
\newblock {\em Handbook of numerical analysis}, 15:89--167.

\bibitem[Livieri et~al., 2018]{livieri2018rough}
Livieri, G., Mouti, S., Pallavicini, A., and Rosenbaum, M. (2018).
\newblock Rough volatility: evidence from option prices.
\newblock {\em IISE Transactions}, 50(9):767--776.

\bibitem[Pagliarani and Pascucci, 2017]{pagliarani2017exact}
Pagliarani, S. and Pascucci, A. (2017).
\newblock The exact {T}aylor formula of the implied volatility.
\newblock {\em Finance and Stochastics}, 21:661--718.

\bibitem[Pigato, 2019]{pigato2019extreme}
Pigato, P. (2019).
\newblock Extreme at-the-money skew in a local volatility model.
\newblock {\em Finance and Stochastics}, 23:827--859.

\bibitem[Rogers, 2023]{rogers2023things}
Rogers, L. (2023).
\newblock Things we think we know.
\newblock In {\em Options—45 years since the Publication of the
Black--Scholes--Merton Model: The Gershon Fintech Center Conference}, pages
173--184. World Scientific.

\bibitem[R{\o}mer, 2022]{romer2022empirical}
R{\o}mer, S.~E. (2022).
\newblock Empirical analysis of rough and classical stochastic volatility
models to the {SPX} and {VIX} markets.
\newblock {\em Quantitative Finance}, 22(10):1805--1838.

\bibitem[Rosenbaum, 2008]{rosenbaum2008estimation}
Rosenbaum, M. (2008).
\newblock Estimation of the volatility persistence in a discretely observed
diffusion model.
\newblock {\em Stochastic Processes and their Applications}, 118(8):1434--1462.

\bibitem[Sarantsev, 2015]{sarantsev2015triple}
Sarantsev, A. (2015).
\newblock Triple and simultaneous collisions of competing {B}rownian particles.
\newblock 20:1--28.

\bibitem[Shi and Yu, 2022]{shi2022volatility}
Shi, S. and Yu, J. (2022).
\newblock Volatility puzzle: Long memory or antipersistency.
\newblock {\em Management Science}.

\bibitem[Wang et~al., 2011]{wang2011collision}
Wang, X., Guo, J., and Jiang, G. (2011).
\newblock Collision local times of two independent fractional {B}rownian
motions.
\newblock {\em Frontiers of Mathematics in China}, 6:325--338.

\end{thebibliography}


\end{document}